%% file: main.tex
\documentclass[article,ij4uq]{ij4uq}      %  review version
\usepackage[hang]{footmisc}
\fancypagestyle{plain}{%
  \fancyhf{}
  \fancyhead[R]{\small {\it}, \thepage--\pageref{LastPage} (\myyear\today)}
  \fancyfoot[R]{\small\bf\thepage }
  \fancyfoot[L]{\fottitle}
  }
\renewcommand{\myyear}{2026}
\renewcommand{\today}{}

\begin{document}

\volume{\myyear\today} 
\title{Uncertainty Quantification in Coupled Multiphysics Systems via Gaussian Process Surrogates: Application to Fuel Assembly Bow}
\titlehead{Uncertainty Quantification in Coupled Multiphysics Systems via GP}
\authorhead{A. Abboud, S. de Lambert, J. Garnier, B. Leturcq}
%For at least  authors with different addresses, use instead the following commands
\corrauthor[1]{Ali Abboud}
\author[2]{Stanislas de Lambert}
\author[1]{Josselin Garnier}
\author[2]{Bertrand Leturcq}
\corremail{ali.abboud@polytechnique.edu}
\corraddress{Centre de Math\'ematiques Appliqu\'ees, Ecole polytechnique, Institut Polytechnique de Paris, 91120 Palaiseau, France}
\address[1]{Centre de Math\'ematiques Appliqu\'ees, Ecole polytechnique, Institut Polytechnique de Paris, 91120 Palaiseau, France}
\address[2]{Université Paris-Saclay, CEA, Service d'Études Mécaniques et Thermiques, 91191, Gif-sur-Yvette, France}

% End information for at least  authors with different addresses
% For authors with the same post address,
%\corrauthor{First A. Author}
%\corremail{f.author@affiliation.com}
%\author{Second B. Author, Jr.}
%\address{Department of Chemistry and Courant, Institute of Mathematical Sciences, New York, NY 10012, USA}
% End commands for all authors with the same address

\dataO{mm/dd/yyyy}
\dataO{}
\dataF{mm/dd/yyyy}
\dataF{}

\abstract{Predicting fuel assembly bow in pressurized water reactors requires solving tightly coupled fluid-structure interaction problems, whose direct simulations can be computationally prohibitive, making large-scale uncertainty quantification (UQ) very challenging. This work introduces a general mathematical framework for coupling Gaussian process (GP) surrogate models representing distinct physical solvers, aimed at enabling rigorous UQ in coupled multiphysics systems. A theoretical analysis establishes that the predictive variance of the coupled GP system remains bounded under mild regularity and stability assumptions, ensuring that uncertainty does not grow uncontrollably through the iterative coupling process. The methodology is then applied to the coupled hydraulic-structural simulation of fuel assembly bow, enabling global sensitivity analysis and full UQ at a fraction of the computational cost of direct code coupling. The results demonstrate accurate uncertainty propagation and stable predictions, establishing a solid mathematical basis for surrogate-based coupling in large-scale multiphysics simulations.}


\maketitle

\section{Introduction}
\label{sec::Intro}
\input{intro}
%------------------------------------------------

\section{General Methodology}

\label{UncertaintySM}
\input{CoupledGPUQ}

\section{Application: Uncertainty quantification of Fuel assembly bow using the Surrogate Models}
\label{SurrogateCoupling}

\input{ApplicationToFABow.tex}

%------------------------------------------------

\section{Conclusion}

\label{Conclusion}

\input{conclusion}

%% The Acknowledgements part is started with the command \acknowledgements;
%% acknowledgements are then done as normal sections before appendix
%% \acknowledgements

%% \acknowledgements

%% This research was supported by the Computational Mathematics program of
%% AFOSR (grant FA9550-07-1-0139).

%% The Appendices part is started with the command \appendix;
%% appendix sections are then done as normal sections and after Acknowledgements
%% \appendix

%% \section{}
%% \label{}

%% References without bibTeX database:

%\begin{thebibliography}{-8}

%% \bibitem must have the following form:

%\small{
%\bibitem{key}

%...

%}

%\end{thebibliography}

%% References with bibTeX database:

\bibliographystyle{IJ4UQ_Bibliography_Style}

\bibliography{references}
\end{document}

%% file: intro.tex
\label{ModelinFSI}

Many multiphysics problems require coupling distinct computational codes, each solving a specific physical subsystem. Such partitioned code coupling schemes iterate between the solvers, exchanging boundary data until a global equilibrium is reached. While conceptually flexible, this approach becomes computationally demanding when initial solvers are involved and many iterations are needed. This challenge is magnified in uncertainty quantification (UQ) and global sensitivity analysis, where thousands of coupled evaluations are typically required. The goal is no longer just to compute a coupled state but to quantify how input uncertainties propagate through an iterative multi-code coupling. To make this tractable, initial solvers can be replaced by probabilistic surrogate models, notably Gaussian process (GP) surrogates \cite{rasmussen2006gaussian}, which provide both predictions and associated uncertainty estimates. Integrating these surrogates in a coupled system yields to a GP-based coupled system, where outputs of one GP feed the inputs of another until convergence. This raises fundamental mathematical questions: how to propagate uncertainty consistently through the coupled system, under what conditions the predictive variance remains bounded, and how to ensure stability and consistency of the coupled predictions. This work addresses these issues by formulating GP coupling as a fixed-point operator on Gaussian measures and proving variance control results under mild regularity assumptions. The methodology is first illustrated on a simple analytical test case with known exact solution and is then applied to a realistic fluid-structure interaction (FSI) problem involving the deformation of fuel assemblies in pressurized water reactors \cite{wanningerphd,delam2021}.

%% file: CoupledGPUQ.tex
This section investigates the propagation of uncertainty arising from the coupling of two independently constructed surrogate models, each representing a distinct physical code. Building on the Bayesian framework of Gaussian Processes (GPs), which inherently provide uncertainty estimates, we develop and analyze a strategy to quantify how these uncertainties interact through the coupling. A combination of theoretical analysis and Monte Carlo simulations is used to validate the approach, culminating in an application to the modeling of fuel assembly bow phenomena in the next section.

\subsection{Coupling of two computational codes as a fixed-point problem}
\label{generalcoupling}

Before addressing uncertainty quantification in the context of coupled surrogate models, it is necessary to give a precise mathematical formulation of the deterministic coupling problem. We present here the abstract framework, introduce the operators involved, and establish that the coupled solution is well defined and that its approximation can be characterized through finitely many evaluations of the solvers.

\medskip
\noindent We consider the coupling of a finite number $C\ge 2$ of computational solvers. For readability, Sections \ref{generalcoupling} \ref{UQGPMC} focus on the practical case $C=2$. The general case $C$ is treated in Section \ref{theoreticalControl} for the theoretical analysis.

\medskip
\noindent\textbf{Interface space.} 
The interaction between the two solvers is mediated through a set of interface variables. These variables are collected in a vector $u$ belonging to the interface space $\mathcal{U}$. In theory,  $\mathcal{U}$ may be viewed as a Banach space, but in practice we only work with finite-dimensional subspaces $\mathcal{U} \subset \mathbb{R}^d$. By construction, $\mathcal{U}$ is distinct from the input spaces $\mathcal{X}_c$ ($c=1,2$) of the individual solvers: it contains only the quantities that are communicated between subsystems. Similarly, the input and output spaces $\mathcal{X}_c, \mathcal{Y}_c$ are finite-dimensional subsets of $\mathbb{R}^{d_c}$ and $\mathbb{R}^{D_c}$, respectively.

\begin{definition}[Computational solvers]
For $c=1,2$, let $\mathcal{S}_c : \mathcal{X}_c \to \mathcal{Y}_c$
be a deterministic solver associated with subsystem $c$. The space $\mathcal{X}_c$ denotes admissible inputs, and $\mathcal{Y}_c$ the corresponding outputs. 
\end{definition}
In practice, each solver can only be accessed through a single evaluation at each iteration.
\begin{definition}[Transfer operators]
The solvers exchange data through transfer operators
\begin{equation}
\Gamma_1 : \mathcal{U}\to\mathcal{X}_1, 
\quad 
\Gamma_{12}:\mathcal{Y}_1\to\mathcal{X}_2,
\quad 
\Gamma_2 : \mathcal{Y}_2\to\mathcal{U}.
\end{equation}
The operator $\Gamma_1$ maps the interface state into input data for solver $\mathcal{S}_1$, $\Gamma_{12}$ transfers the output of $\mathcal{S}_1$ to an admissible input for solver $\mathcal{S}_2$, and $\Gamma_2$ updates the interface state from the output of $\mathcal{S}_2$.
\end{definition}

\begin{definition}[Coupling operator] 
The global coupling operator is defined as 
\begin{equation}
\mathcal{T} := \Gamma_2 \circ \mathcal{S}_2 \circ \Gamma_{12} \circ \mathcal{S}_1 \circ \Gamma_1 : \mathcal{U} \to \mathcal{U}.
\end{equation}
Given an interface state $u\in\mathcal{U}$, one application of $\mathcal{T}$ consists of: mapping $u$ into an input for solver $\mathcal{S}_1$ via $\Gamma_1$, running $\mathcal{S}_1$ once, transferring its output to $\mathcal{X}_2$ using $\Gamma_{12}$, running $\mathcal{S}_2$ once, and mapping the result back into $\mathcal{U}$ with $\Gamma_2$. A coupled solution is a fixed point $u^\star = \mathcal{T}(u^\star),$
i.e.\ an interface state consistent with both solvers.
\end{definition}

\begin{proposition}[Banach fixed-point theorem \cite{zeidler1986nonlinear}]
\label{prop:banach}
Assume $\mathcal{T}$ is a contraction on a Banach space $(\mathcal{U},\|\cdot\|_\mathcal{U}x)$. Then, for any initial $u^{(0)}\in\mathcal{U}$, the sequence 
\begin{equation}
u^{(m+1)} := \mathcal{T}(u^{(m)}), \quad m\geq 0,
\end{equation}
converges to the unique fixed point $u^\star$.
\end{proposition}

\begin{definition}[Iterative termination]
In computations, the iteration is stopped at index $M<\infty$ when the criterion 
\begin{equation}
\|u^{(M)}-u^{(M-1)}\|_\mathcal{U} \leq \varepsilon
\end{equation}
is satisfied for a prescribed tolerance $\varepsilon>0$.
\end{definition}

\medskip
\noindent\textbf{Finite evaluations.} 
Since each solver $\mathcal{S}_c$ is called exactly once per iteration, we denote by
$x^{(c)}_m \in \mathcal{X}_c$ the input passed to solver $\mathcal{S}_c$ at iteration $m$, for $c=1,2$.
After $M$ iterations, the total number of evaluations of each solver is therefore $M$.
Thus, the approximation $u^{(M)}$ is entirely determined by the finite sets of evaluations 
\begin{equation}
\label{eq:finiteset}
\big\{\mathcal{S}_c(x^{(c)}_m) : m=1,\dots,M , \, c=1,2\big\}.
\end{equation}

\medskip
\noindent\textbf{Finite characterization of the coupled approximation.}
The computed approximation $u^{(M)}$ depends only on these finitely many evaluations. Moreover, by Proposition \ref{prop:banach}, $u^{(M)} \to u^\star$ as $M\to\infty$, so the fixed point $u^\star$ is approximated by iterates $u^{(M)}$, each of which involves only finitely many evaluations.

\medskip
\noindent\emph{Proof of (\ref{eq:finiteset}).} We proceed by induction on $m$. At $m=0$, $u^{(0)}$ is prescribed. The first update is
\begin{equation}
u^{(1)}=\Gamma_2\big(\mathcal{S}_2(\Gamma_{12}(\mathcal{S}_1(\Gamma_1(u^{(0)}))))\big),
\end{equation}
which requires one evaluation of $\mathcal{S}_1$ at $x^{(1)}_1=\Gamma_1(u^{(0)})\in\mathcal{X}_1$ and one evaluation of $\mathcal{S}_2$ at $x^{(2)}_1=\Gamma_{12}(\mathcal{S}_1(x^{(1)}_1))\in\mathcal{X}_2$. 

Suppose now that $u^{(m)}$ depends only on the finite collection of evaluations
\begin{equation}
\{\mathcal{S}_1(x^{(1)}_{m'}): m'=1,\dots,m\}, 
\qquad 
\{\mathcal{S}_2(x^{(2)}_{m'}): m'=1,\dots,m\}.
\end{equation}
Then computing $u^{(m+1)}=\mathcal{T}(u^{(m)})$ requires evaluating $\mathcal{S}_1$ at a new input $x^{(1)}_{m+1}$ determined from $u^{(m)}$, and subsequently $\mathcal{S}_2$ at a new input $x^{(2)}_{m+1}=\Gamma_{12}(\mathcal{S}_1(x^{(1)}_{{m+1}}))$. Thus $u^{(m+1)}$ depends only on the previous evaluations together with these two new ones.  

By induction, for each $m$ the interface variable $u^{(m)}$ is determined by finitely many evaluations of the form $\mathcal{S}_c(x^{(c)}_m)$. Since the algorithm terminates after $M<\infty$ iterations, $u^{(M)}$ depends only on finitely many such evaluations.
\qed

The coupled solution is well defined as a fixed point of the operator $\mathcal{T}$, and its approximation reduces to a finite set of evaluations of the solvers at points $x^{(c)}_m$. This finite characterization is the foundation for introducing surrogate models in place of $\mathcal{S}_c$, and for developing uncertainty quantification in the next subsection.

\subsection{Uncertainty Quantification in Coupled Gaussian Process Surrogates}

\label{UQGPMC}

In the context of coupling physical codes, replacing computationally expensive numerical solvers with surrogate models-typically Gaussian Processes (GPs)-offers significant performance gains. In general cases, if the surrogates show strong predictive capabilities, it can justify the use of their mean predictions to replace the original codes for UQ methodology. However, these Gaussian Process surrogates inherently carry epistemic uncertainty due to their probabilistic nature. Assuming that the mean prediction alone is sufficient
in Coupled Multiphysics Systems via GP neglects this uncertainty, which may propagate through the coupled simulation workflow and impact the accuracy and robustness of the final results. To study the propagation of this uncertainty in this context we describe a Monte Carlo strategy below.

\textbf{Notation.}
For $c=1,2$, $f^{(c)}$ denotes the surrogate of solver $\mathcal{S}_c$ with input space $\mathcal{X}^{(c)}\subset\mathbb{R}^{d_c}$. The outputs of $f^{(c)}$ may be vector-valued in $\mathbb{R}^{D_c}$, in which case the covariance kernel is matrix-valued and $f^{(c)}$ is a (vector-valued) multi-output GP \cite{alvarez2012,micchelli2005}. Let $\mathbf{y}$ denote the vector of coupled outputs returned by the fixed-point algorithm (it may be equal to the interface variable $u$ or be a function of it,  we keep the notation $\mathbf{y}$ to stress vector-valued quantities).

Throughout, we assume an isotopic design for each surrogate $c$ (the same input set is used for all output coordinates).

\textbf{Correlation structures for multi-output GPs.}
A multi-output (vector-valued) Gaussian process (MOGP) is a random field $x \in \mathcal{X}^{(c)} \mapsto f^{(c)}(x)\in\mathbb{R}^{D_c}$ with mean $\mu^{(c)}(x)\in\mathbb{R}^{D_c}$. Writing 
\begin{equation}
f^{(c)}(x)=\big(f^{(c)}_1(x),\dots,f^{(c)}_{D_c}(x)\big),
\label{eq:fk-decomposition}
\end{equation}
the process $\{f^{(c)}_\ell(x)\}_{\ell,x}$ is Gaussian if and only if, for any $n\ge 1$, any inputs $x_1,\dots,x_n \in \mathcal{X}^{(c)}$ and any indices $\ell_1,\dots,\ell_n\in\{1,\dots,D_c\}$, the linear combination
$
\sum_{j=1}^n \lambda_j\, f^{(c)}_{\ell_j}(x_j)
$
has a univariate normal distribution. Defining the cross-covariance functions
\begin{equation}
\kappa^{(c)}_{\ell,\ell'}(x,x') := \mathrm{Cov}\big(f^{(c)}_\ell(x),f^{(c)}_{\ell'}(x')\big),
\label{eq:cross-cov}
\end{equation}
we can collect them into a matrix-valued covariance kernel
\begin{equation}
\kappa^{(c)}:\ \mathcal{X}^{(c)}\times\mathcal{X}^{(c)}\to\mathbb{R}^{D_c\times D_c},
\qquad
\kappa^{(c)}(x,x') = \big[\kappa^{(c)}_{\ell,\ell'}(x,x')\big]_{\ell,\ell'=1}^{D_c}.
\label{eq:matrix-kernel}
\end{equation}
For any finite design set $X=\{x_i\}_{i=1}^n$, stacking $(f^{(c)}(x_1)^\top,\dots,f^{(c)}(x_n)^\top)^\top\in\mathbb{R}^{nD_c}$ yields a block covariance matrix $K^{(c)}(X,X)$ with entries
\begin{equation}
K^{(c)}_{(i,\ell),(j,\ell')} = \kappa^{(c)}_{\ell,\ell'}(x_i,x_j).
\label{eq:block-entries}
\end{equation}
The matrix $K^{(c)}(X,X)$ is symmetric and positive semi-definite (PSD), i.e.
\begin{equation}
v^\top K^{(c)}(X,X)\,v \ge 0 \qquad \text{for all } v\in\mathbb{R}^{nD_c}.
\label{eq:psd-def}
\end{equation}
Thus, specifying a MOGP amounts to choosing cross-covariance functions $\kappa^{(c)}_{\ell,\ell'}$ such that the associated block covariance matrices are PSD for all finite sets of inputs. Two standard choices are:
\begin{itemize}
\item \emph{Independent outputs.} Take a diagonal kernel
\begin{equation}
\kappa^{(c)}(x,x')=\mathrm{diag}\!\big(\kappa^{(c)}_1(x,x'),\dots,\kappa^{(c)}_{D_c}(x,x')\big),
\label{eq:indep-kernel}
\end{equation}
so the $D_c$ coordinates are a priori uncorrelated, i.e. $\kappa^{(c)}_{\ell,\ell'}(x,x')=0$ for $\ell\neq \ell'$. For an isotopic design (same inputs for all outputs), the joint covariance is block diagonal
\begin{equation}
K^{(c)}(X,X) = \mathrm{diag}\!\big(K^{(c)}_1,\dots,K^{(c)}_{D_c}\big)
\quad\text{with}\quad
\big[K^{(c)}_\ell\big]_{ij}=\kappa^{(c)}_\ell(x_i,x_j).
\label{eq:block-diag}
\end{equation}

\item \emph{Structured correlations (Linear Model of Coregionalization, LMC).} Let $Q\ge 1$ and set
\begin{equation}
\kappa^{(c)}(x,x')=\sum_{q=1}^{Q} B^{(c)}_{q}\,\kappa^{(c)}_{q}(x,x'),
\label{eq:lmc-kernel}
\end{equation}
where each $B^{(c)}_{q}\in\mathbb{R}^{D_c\times D_c}$ is a PSD coregionalization matrix and $\kappa^{(c)}_{q}$ is a scalar covariance on $\mathcal{X}^{(c)}$ \cite{bonilla2008,alvarez2012}. Equivalently, writing $B^{(c)}_{q}=A^{(c)}_{q}(A^{(c)}_{q})^\top$ shows that outputs are %instantaneous 
linear mixtures of $Q$ independent latent GPs with kernels $\kappa^{(c)}_{q}$, which guarantees PSD of the full block covariance. In the isotopic case, the joint covariance over all outputs and all inputs reads
\begin{equation}
%\mathbf{K}^{(c)}
K^{(c)}(X,X)=\sum_{q=1}^Q B^{(c)}_{q}\otimes K^{(c)}_{q},\qquad \big[K^{(c)}_{q}\big]_{ij}=\kappa^{(c)}_{q}(x_i,x_j),
\label{eq:lmc-joint-K}
\end{equation}
making explicit the separation between output-output dependence (via $B^{(c)}_{q}$) and input-input dependence (via $K^{(c)}_{q}$).
In our application we adopt independent outputs (diagonal $B^{(c)}_{q}$), but we emphasize that the methodology accommodates LMC without further changes.
\end{itemize}

\textbf{Surrogate models.}
Each surrogate $f^{(c)}$ is trained on evaluations of $\mathcal{S}_c$ from a space-filling design (e.g., Latin hypercube or low-discrepancy sampling \cite{damblin2013numerical}). A priori we have
$$
f^{(c)}(x)\sim \mathcal{GP}\big(\mu^{(c)}(x),\,\kappa^{(c)}(x,x')\big),
$$
where $\mu^{(c)}(x)\in\mathbb{R}^{D_c}$ and, in the multi-output case, $\kappa^{(c)}(x,x')\in\mathbb{R}^{D_c\times D_c}$ is matrix-valued. 
If $\bar{\mu}^{(c)}$ denotes the posterior mean and
 $\bar{\kappa}^{(c)}$ denotes the posterior kernel, then for each $x$ we have 
$\mathbf{0}\ \preceq\ \bar{\kappa}^{(c)}(x,x)\ \preceq\ \kappa^{(c)}(x,x),$
with $\preceq$ the Loewner order on symmetric PSD matrices. In the independent-output case, this reduces componentwise to $0\le \bar{\kappa}^{(c)}_\ell(x,x)\le \kappa^{(c)}_\ell(x,x)$.

\textbf{Monte Carlo framework for uncertainty quantification of the coupled GP.}
Uncertainty quantification is performed by a standard Monte Carlo procedure. An outer loop generates i.i.d.\ realizations of the coupled solution: for each iteration $j=1, \ldots, N$, independent realizations $\left\{f^{(c, j)}\right\}_c$ of the GP surrogates are drawn, the coupling algorithm is executed with these surrogates starting from a prescribed initial interface state $u^{(0)}$, and the corresponding coupled solution $\mathbf{y}^{(j)}$ is obtained. We denote by $M^{(j)}$ the number of fixed-point iterations performed for simulation $j$ before the stopping criterion $\|u^{(m,j)}-u^{(m-1,j)}\|_{\mathcal{U}} \leq \varepsilon$
is met. In general, $M^{(j)}$ depends on the Monte Carlo realization $j$. The sample $\left\{\mathbf{y}^{(j)}\right\}_{j=1}^N$ is then used to compute empirical statistics such as means, variances, and confidence intervals.

Three approaches can in principle be used to generate the surrogate realizations inside the coupling loop:
\begin{itemize}
    \item \textbf{Ideal but infeasible:}  For each Monte Carlo replication $j$, a full realization $f^{(c,j)}$ of the GP surrogate $f^{(c)}$ over the entire input space $\mathcal{X}_c$ is drawn. The coupling algorithm is then run with $f^{(c,j)}$. This would provide exact i.i.d.\ samples of the coupled solution, but it is computationally impossible since $\mathcal{X}_c$ is infinite.

       \item \textbf{Rigorous but computationally heavy:}
       For each Monte Carlo replication $j$, the coupled solver is run with a fully stochastic surrogate
    family $\{f^{(c,j)}\}_c$ that is sampled consistently with the posterior GP laws.
    More precisely, at iteration $m$ of replication $j$, when the coupling algorithm queries surrogate $c$
    at a new input $x^{(c,j)}_m\in\mathcal{X}^{(c)}$, let
   $$
      X_{m-1}^{(c,j)} := \big(x^{(c,j)}_1,\dots,x^{(c,j)}_{m-1}\big), 
      \qquad 
      Z_{m-1}^{(c,j)} := \big(f^{(c,j)}(x^{(c,j)}_1),\dots,f^{(c,j)}(x^{(c,j)}_{m-1})\big)
   $$
    denote the locations and values already sampled of the surrogate $c$  along the trajectory of replication $j$.
    By consistency of the GP posterior, the conditional distribution of the vector $f^{(c)}(x^{(c,j)}_m)$
    given $f^{(c)}(X_{m-1}^{(c,j)})=Z_{m-1}^{(c,j)}$
    is multivariate Gaussian with explicitly computable mean and variance.
    The value $f^{(c,j)}(x^{(c,j)}_m)$ is then drawn from this conditional law and appended to $Z_{m-1}^{(c,j)}$,
    and the coupled fixed-point iteration is updated using these sampled values.
    This construction ensures that, for each $j$, the random trajectory $\{f^{(c,j)}(x^{(c,j)}_m)\}_{m,c}$ is
    exactly distributed according to the joint GP posterior along the (random) coupling path, so that the
    outputs $\mathbf{y}^{(j)}$ are rigorous samples of the coupled GP-based model.
    However, at each new query the conditional Gaussian formulas involve covariance matrices of size $m\times m$
    built at locations $\{x^{(c,j)}_1,\dots,x^{(c,j)}_m\}$ which depend on both $m$ and $j$, so that repeated inversions of growing matrices are required. The resulting cost is very high, but the
    procedure is fully coherent with the GP prior and posterior, and can serve, at least on moderate grids, as a
    reference method for the example of Section \ref{analyticalExample}.

    \item \textbf{Proposed method:} The proposed method is a practical compromise described below, which fixes the input sets after one deterministic coupling run and reuses them for all MC iterations, thereby greatly reducing the cost while retaining a non-trivial description of uncertainty.

We first fix an initial interface state $u^{(0)}$ and run the deterministic coupling using the GP posterior means $\bar{\mu}^{(c)}$. This fixed-point iteration produces a sequence $(u^{(m)})_{m\geq 0}$ and, by the stopping criterion, converges after $M$ iterations. The points called during this deterministic computation at the $m$-th iteration are denoted
$$
    x_m^{(c)}\in\mathcal{X}^{(c)},\qquad m=1,\dots,M,\ \ c=1,2.
$$
These same points $\{x_m^{(c)}\}_{m=1}^M$ are reused for all Monte Carlo replications, which keeps the sampling cost low and fixes a unique reference path starting from $u^{(0)}$.

For each Monte Carlo replication $j=1,\dots,N$ and for each $c$, we draw the Gaussian vector
$$
\big(\tilde{f}^{(c,j)}(x_m^{(c)})\big)_{m=1}^{M}
\ \sim\ \mathcal{N}\!\big(\bar\mu^{(c)}(x^{(c)}),\,\bar{\mathbf{K}}^{(c)}\big),
$$
that is, a realization of the vector $\big(f^{(c)}(x_m^{(c)})\big)_{m=1}^{M}$. Here $\bar\mu^{(c)}(x^{(c)})=\big(\bar\mu^{(c)}(x^{(c)}_m)\big)_{m=1}^{M}$ and $\bar{\mathbf{K}}^{(c)}$ is the block covariance matrix whose $(m,m')$ block is $\bar\kappa^{(c)}\!\big(x_m^{(c)},x_{m'}^{(c)}\big)\in\mathbb{R}^{D_c\times D_c}$.

We then define the errors
\begin{equation}
  \delta_m^{(c,j)}=\tilde{f}^{(c,j)}(x_m^{(c)})-\bar\mu^{(c)}(x_m^{(c)}),\qquad m=1,\dots,M
\end{equation}
and, at simulation $j$ and iteration $m$, we use the perturbed surrogate
\begin{equation}
\hat f^{(c,m,j)}(x)=\bar\mu^{(c)}(x)+\delta_m^{(c,j)}.
\end{equation}

Running the coupling with $\{\hat f^{(c,m,j)}\}$, starting from the same initial state $u^{(0)}$, yields a vector output $\mathbf{y}^{(j)}$. We can anticipate that the i.i.d. family $\{\mathbf{y}^{(j)}\}_{j=1}^N$ obtained by this procedure has the same distribution as the one of the rigorous but computationally heavy procedure described above, provided the sampled points $x^{(c,j)}_m$ stay close enough (relative to the correlation radius) to $x^{(c)}_m$. This will be quantified in the next subsection.

From the family $\{\mathbf{y}^{(j)}\}_{j=1}^N$, we compute
\begin{equation}
\hat{\mathbf{y}}=\frac{1}{N}\sum_{j=1}^N \mathbf{y}^{(j)},\qquad
\widehat{\mathrm{Cov}}[\mathbf{y}]=\frac{1}{N-1}\sum_{j=1}^N\big(\mathbf{y}^{(j)}-\hat{\mathbf{y}}\big)\big(\mathbf{y}^{(j)}-\hat{\mathbf{y}}\big)^\top.
\end{equation}
When needed, componentwise summaries (means, variances, confidence intervals) are obtained from $\hat{\mathbf{y}}$ and the diagonal of $\widehat{\mathrm{Cov}}[\mathbf{y}]$,  these quantities describe the uncertainty on the coupled outputs $\mathbf{y}$.

\end{itemize}

\subsection{Theoretical Control of Predictive Variance}
  
\label{theoreticalControl}
\noindent\textbf{Objective and strategy.}
The purpose of this subsection is to provide a theoretical justification for the Monte Carlo perturbation scheme of Section \ref{UQGPMC}, which is used to propagate Gaussian process (GP) uncertainty through coupled surrogate models. More precisely, the aim is to show that the random coupled outputs produced by the practical algorithm are concentrated around the deterministic fixed point obtained with GP mean surrogates, and that their dispersion is explicitly controlled by the GP posterior variances via finite-sample bounds on the effect of surrogate uncertainty on the coupled solution.

The argument proceeds in three steps: First, relate the posterior variances of scalar and multi-output GPs to the fill distance of the design points, by exploiting the kernel-based approximation results of Kanagawa et al. \cite{kanagawa2018}, which connect GP posterior variance to worst-case interpolation error in the associated reproducing kernel Hilbert space (RKHS). In particular, Matérn kernels induce Sobolev RKHSs, so that the variance can be controlled through the fill distance of a representative input set. Second, establish a Lipschitz stability result for the coupled fixed point with respect to perturbations of the surrogates. And finally, combine these two ingredients to obtain a finite-sample, high-probability bound (Proposition \ref{prop:finitesampleHQBound_LMC_corrected}), which constitutes the main theoretical result of this section and quantifies the uncertainty of the Monte Carlo outputs.

\medskip
\noindent\textbf{Number of coupled solvers.}
In Sections \ref{generalcoupling}-\ref{UQGPMC}, the presentation focused on the case of two coupled solvers ($C=2$), which corresponds to the application considered in this work.
In the present subsection, the notation is slightly generalized and the analysis is carried out for a generic finite number $C\ge 2$ of coupled surrogates,  the case $C=2$ is recovered by specializing the formulas to two components.

\medskip
\noindent\textbf{Vector-valued convention.}
In coherence with the previous subsection, outputs (including ${\bf y}$) are vectors and GP kernels may be matrix-valued. Accordingly, covariance inequalities are to be understood in Loewner (PSD) order and pointwise deviations are measured with the Euclidean norm. When we write a sup-norm for vector-valued functions $g:\mathcal{X}\to\mathbb{R}^{D}$, we mean the sup-Euclidean norm
$\|g\|_{\infty,2}:=\sup_{x\in\mathcal{X}}\|g(x)\|_2$.

A first ingredient is a geometric quantity that characterizes how well a design $X$ covers the input domain and will later control the decay of the GP posterior variance.

\begin{definition}[Fill distance]
 The fill distance of a point set $X \subset \mathcal{X}$ with respect to a domain $\mathcal{X} \subset \mathbb{R}^d$ is defined as 
\begin{equation}
h_{X,\mathcal{X}} := \sup_{x \in \mathcal{X}} \min_{x_i \in X} \|x - x_i\|,
\end{equation}
which measures the largest gap between any point in the domain $\mathcal{X}$ and its nearest representative in $X$.
\end{definition}

The following result recalls the standard multi-output GP regression formulas and fixes notation for the block covariance matrices used in the subsequent bounds.

\begin{theorem}[Posterior mean and covariance for multi-output GP regression {\cite{micchelli2005,alvarez2012}}]
Let $f:\mathcal{X}\to\mathbb{R}^{D}$ have a Gaussian process prior $\mathfrak{f}\sim\mathcal{GP}(\mu,\kappa)$ with mean function $\mu:\mathcal{X}\to\mathbb{R}^{D}$ and matrix-valued kernel $\kappa:\mathcal{X}\times\mathcal{X}\to\mathbb{R}^{D\times D}$ (positive semidefinite in the sense of Loewner order). Given inputs $X=(x_1,\dots,x_n)\in\mathcal{X}^n$ and noisy vector-valued observations
$\mathbf{Z}=\big(z(x_1)^\top,\dots,z(x_n)^\top\big)^\top\in\mathbb{R}^{nD},\qquad 
\mathbf{Z}=f(X)+\varepsilon,\ \ \varepsilon\sim\mathcal{N}(0,\sigma^2 I_{nD}),$ 
the posterior is $\mathfrak{f}\mid \mathbf{Z}\sim\mathcal{GP}(
%\mu_{\mathrm{post}}
\bar\mu,\bar{\kappa})$ with
\begin{equation}
%\mu_{\mathrm{post}}
\bar\mu (x)\;=\;\mu(x)\;+\;K_{xX}\,\big(K_{XX}+\sigma^2 I_{nD}\big)^{-1}\,\big(\mathbf{Z}-\mu_X\big),
\end{equation}
\begin{equation}
\bar{\kappa}(x,x')\;=\;\kappa(x,x')\;-\;K_{xX}\,\big(K_{XX}+\sigma^2 I_{nD}\big)^{-1}\,K_{Xx'},
\end{equation}
where $\mu_X=\big(\mu(x_1)^\top,\dots,\mu(x_n)^\top\big)^\top\in\mathbb{R}^{nD}$,
$K_{XX}\in\mathbb{R}^{nD\times nD}$ is the block covariance with $(i,j)$-block $[K_{XX}]_{ij}=\kappa(x_i,x_j)\in\mathbb{R}^{D\times D}$,
$K_{xX}=\big[\kappa(x,x_1)\ \cdots\ \kappa(x,x_n)\big]\in\mathbb{R}^{D\times nD}$, and $K_{Xx'}=K_{x'X}^\top$.
In the noise-free case, replace $K_{XX}+\sigma^2 I_{nD}$ by $K_{XX}$.
\end{theorem}

In order to relate GP posterior variances to approximation properties in Sobolev spaces, a classical characterization of Matérn kernels in terms of their associated RKHS is recalled next.

\begin{proposition}[RKHS of Matérn kernels \cite{kanagawa2018}]
Let $\kappa_{\alpha,h}$ be the Matérn kernel on $\mathcal{X}\subset\mathbb{R}^{d_{c}}$ with smoothness $\alpha>0$ and scaling $h>0$. If $s=\alpha+d_{c}/2$ is an integer, the RKHS $\mathcal{H}_{\kappa_{\alpha,h}}$ is norm-equivalent to the Sobolev space $W_2^{s}(\mathcal{X})$:
$$W_2^{s}(\mathcal{X})=\Big\{f\in L_2(\mathcal{X}): \sum_{|\beta|\leq s}\|D^\beta f\|_{L_2(\mathcal{X})}^2 < \infty\Big\}.$$
Thus $\mathcal{H}_{\kappa_{\alpha,h}}=W_2^{s}(\mathcal{X})$ as sets of functions, and there exist constants $c_1,c_2>0$ with
$c_1\|f\|_{W_2^{s}(\mathcal{X})} \leq \|f\|_{\mathcal{H}_{\kappa_{\alpha,h}}} \leq c_2\|f\|_{W_2^{s}(\mathcal{X})}, 
\quad \forall f\in \mathcal{H}_{\kappa_{\alpha,h}}.$
\label{assemp2}
\end{proposition}

Under these regularity assumptions, the next result of Kanagawa et al. \cite{kanagawa2018} links the posterior variance of a scalar GP to the local fill distance of the design.

\begin{theorem}[\emph{Scalar} fill-distance bound,  Kanagawa et al. {\cite{kanagawa2018}}]
\label{thm:kanagawa_scalar}
Let $\kappa$ be a scalar kernel on $\mathbb R^{d_c}$ whose RKHS is norm-equivalent to the Sobolev space $W_2^{s}(\mathbb R^{d_c})$ with $s>{d_c}/{2}$. For $\rho>0$ and any finite set of points $X=\{x_1,\dots,x_n\}\subset\mathbb R^{d_c}$, define the (local) fill distance at $x\in\mathbb R^{d_c}$ by
\begin{equation}
h_{\rho,X}(x)
\;:=\;
\sup_{\substack{x'\in\mathbb R^{d_c}\\ \|x'-x\|\le \rho}} \ \min_{x_i\in X} \|x'-x_i\|.
\end{equation}
Then, for any $\rho>0$, there exist $h_0>0$ and $\mathcal{C}>0$ such that for all $x\in\mathbb R^{d_c}$ and all point sets $X$ with $h_{\rho,X}(x)\le h_0$, the posterior variance satisfies
\begin{equation}
\bar{\kappa}(x,x)\ \le\ \mathcal{C}\, h_{\rho,X}(x)^{\,2s-d_c}.
\end{equation}
\end{theorem}

This scalar variance bound can be extended to the multi-output setting by exploiting the LMC representation of the matrix-valued kernel.

\begin{theorem}[Using LMC representation a lift of Kanagawa’s scalar bound to multi-output GPs]
\label{thm:lmc_fill}
Let the matrix-valued kernel be of LMC form
$\kappa(x,x')= \sum_{q=1}^{Q} B_q\,\kappa_q(x,x'),$
with $B_q\succeq \mathbf 0$ and each scalar $\kappa_q$ satisfying Theorem \ref{thm:kanagawa_scalar} with parameters $(s_q,\mathcal{C}_q,h_{0,q})$.
Then, for every $x$ with $h_{\rho,X}(x)\le \min_q h_{0,q}$,
\begin{equation}
\bar{\kappa}(x,x)\ \preceq\ \Big(\sum_{q=1}^{Q} \mathcal{C}_q\,h_{\rho,X}(x)^{\,2s_q-d_c}\Big)\,\Big(\sum_{q=1}^{Q} B_q\Big).
\end{equation}

\end{theorem}

\begin{proof}
Fix $x\in\mathcal X$ and $v\in\mathbb R^D$. The scalar projection $g_v(x):=v^\top f(x)$ is a GP with kernel
$
\kappa_v(x,x')=v^\top \kappa(x,x')\,v=\sum_{q=1}^{Q} a_q\,\kappa_q(x,x'),
$
where $a_q:=v^\top B_q v\ge 0$  \cite{micchelli2005vv,alvarez2012review}. 
By Gaussian conditioning monotonicity 
\begin{equation}
    \label{eq:firstdisplay}
v^\top \bar{\kappa}(x,x)\,v \;=\; \mathrm{Var} \!\big(v^\top f(x)\mid f(X)\big)\ \le\ \mathrm{Var} \!\big(v^\top f(x)\mid v^\top f(X)\big)\;=\;\overline{\kappa_v}(x,x).
\end{equation}

Let $\varphi_q$ be feature maps for the $\kappa_q$ so that $\kappa_q(x,x')=\langle\varphi_q(x),\varphi_q(x')\rangle_{\mathcal H_q}$. A feature map for $\kappa_v$ is $\varphi=(\sqrt{a_q}\varphi_q)_q$ into the Hilbert direct sum $\bigoplus_q\mathcal H_q$ \cite{aronszajn1950,alvarez2012review}, since
$$
\big\langle\varphi(x),\varphi(x')\big\rangle_{\oplus_q\mathcal H_q}
\;=\;\sum_{q=1}^Q \big\langle \sqrt{a_q}\varphi_q(x),\,\sqrt{a_q}\varphi_q(x')\big\rangle_{\mathcal H_q}
\;=\;\sum_{q=1}^Q a_q\,\kappa_q(x,x')
\;=\;\kappa_v(x,x').
$$
For this scalar GP with kernel $\kappa_v$, the posterior variance admits the usual RKHS interpretation as the squared distance from $\varphi(x)$ to the span $S:=\mathrm{span}\{\varphi(x_i)\}$ in feature space   \cite{kanagawa2018}:
$$
\sqrt{\overline{\kappa_v}(x,x)}\;=\;\mathrm{dist}\big(\varphi(x),S\big)
\ \le\ \sum_{q=1}^{Q}\sqrt{a_q}\,\mathrm{dist}\big(\varphi_q(x),S_q\big)
\ =\ \sum_{q=1}^{Q}\sqrt{a_q}\,\sqrt{\overline{\kappa_q}(x,x)},
$$
where $S_q:=\mathrm{span}\{\varphi_q(x_i)\}$ and we used the triangle inequality in the direct sum. Squaring and applying Cauchy-Schwarz yields
$$
\overline{\kappa_v}(x,x)\ \le\ \Big(\sum_{q=1}^{Q} \sqrt{a_q}\,\sqrt{\overline{\kappa_q}(x,x)}\Big)^2
\ \le\ \Big(\sum_{q=1}^{Q} a_q\Big)\,\Big(\sum_{q=1}^{Q} \overline{\kappa_q}(x,x)\Big).
$$
By Theorem \ref{thm:kanagawa_scalar}, $\overline{\kappa_q}(x,x)\le \mathcal{C}_q\,h_{\rho,X}(x)^{\,2s_q-d_c}$ whenever $h_{\rho,X}(x)\le h_{0,q}$. Therefore
$$
\overline{\kappa_v}(x,x)\ \le\ \Big(\sum_{q=1}^{Q} a_q\Big)\,\Big(\sum_{q=1}^{Q} \mathcal{C}_q\,h_{\rho,X}(x)^{\,2s_q-d_c}\Big)
\ =\ v^\top\!\Big(\sum_{q=1}^{Q} B_q\Big)v\ \cdot\ \Big(\sum_{q=1}^{Q} \mathcal{C}_q\,h_{\rho,X}(x)^{\,2s_q-d_c}\Big) .
$$
Combining with (\ref{eq:firstdisplay}) gives, for all $v$,
$$
v^\top \bar{\kappa}(x,x)\,v\ \le\ v^\top\!\Big(\sum_{q=1}^{Q} B_q\Big)v\ \cdot\ \Big(\sum_{q=1}^{Q} \mathcal{C}_q\,h_{\rho,X}(x)^{\,2s_q-d_c}\Big),
$$
which is equivalent to the stated Loewner bound.
\end{proof}

\begin{remark}[Interpretation of Theorem \ref{thm:lmc_fill}]
Theorem \ref{thm:lmc_fill} states that, for a multi-output GP with an LMC kernel, the posterior covariance at any point $x$ is bounded above by a scalar factor depending on the fill distance $h_{\rho,X}(x)$, times a fixed output-output matrix $\sum_q B_q$.
In other words, as the design becomes denser (the fill distance decreases), the posterior uncertainty on all outputs decays at a rate determined by the smoothness parameters $s_q$.
This result lifts the scalar bound of Kanagawa et al. \cite{kanagawa2018} to the multi-output setting relevant for our coupled surrogates.
\end{remark}

After controlling the posterior variances of the individual surrogates, the second ingredient concerns the stability of the coupled fixed point with respect to perturbations of these surrogates.

\textbf{Coupling operator and solution map.}
Let $\mathcal U$ be the interface space and $H:\mathcal U\to\mathcal Y$ the post-processing map to outputs.
Fix an integer $C\ge 2$, and let $f=(f^{(1)},\ldots,f^{(C)})$ be a $C$-tuple of surrogates, with fixed interface operators
$$\Gamma_{1}:\mathcal U\to\mathcal X_1, \Gamma_{r,r+1}:\mathcal Y_r\to\mathcal X_{r+1},\qquad r=1,\dots,C-1,
$$
and $\Gamma_{C}:\mathcal Y_C\to\mathcal U$,
define the %parametric 
coupling operator
$$
\mathcal T_f \;:=\; \Gamma_{C}\circ f^{(C)}\circ \Gamma_{C-1,C}\circ \cdots \circ f^{(1)}\circ \Gamma_{1} \;:\; \mathcal U\to\mathcal U.
$$
We then define the solution map $\Phi:\mathcal A\to\mathcal Y,\qquad
\Phi(f)\ :=\ H\big(u^\star(f)\big),\quad u^\star(f)\in\mathrm{Fix}(\mathcal T_f),$
where $\mathcal A$ is an admissible class of surrogate tuples on which the fixed point exists and is unique.

This stability analysis is based on a uniform contraction property of the coupling operator together with a Lipschitz condition on the post-processing map $H$.

\begin{assumption}[Uniform contraction]
\label{ass:contract}
There exists $\rho\in(0,1)$ such that for all $f\in\mathcal A$, the operator $\mathcal T_f$ is a contraction on a Banach space $(\mathcal U,\|\cdot\|_{\mathcal U})$ with modulus at most $\rho$, and $H$ is $L_H$-Lipschitz.
Moreover, there are constants $L_1,\dots,L_C$ (depending only on the interface operators and uniform Lipschitz bounds of the admissible surrogates) such that for any $f,g\in\mathcal A$ and any $u\in\mathcal U$,
\begin{equation}
\label{eq:one-block}
\big\|\mathcal T_f(u)-\mathcal T_g(u)\big\|_{\mathcal U}
\ \le\ \sum_{c=1}^C L_c\,\big\|f^{(c)}-g^{(c)}\big\|_{\infty,2}.
\end{equation}
\end{assumption}

Under the above contraction assumption, the following proposition shows that the coupled solution depends Lipschitz-continuously on the underlying surrogates.

\begin{proposition}[Lipschitz stability of the solution map]
\label{prop:LipPhi}
Under Assumption \ref{ass:contract}, the map $\Phi$ is Lipschitz on $\mathcal A$:
\begin{equation}
\big\|\Phi(f)-\Phi(g)\big\|_{\mathcal Y}
\ \le\ \frac{L_H}{1-\rho}\,\sum_{c=1}^C L_c\,\big\|f^{(c)}-g^{(c)}\big\|_{\infty,2}
\qquad\text{for all } f,g\in\mathcal A.
\end{equation}
\end{proposition}

\begin{proof}
For any $f,g\in\mathcal A$,
$$
\|u^\star(f)-u^\star(g)\|_{\mathcal U}
=\|\mathcal T_f(u^\star(f))-\mathcal T_g(u^\star(g))\|_{\mathcal U}
\le \|\mathcal T_f(u^\star(f))-\mathcal T_f(u^\star(g))\|_{\mathcal U}
+ \|\mathcal T_f(u^\star(g))-\mathcal T_g(u^\star(g))\|_{\mathcal U}.
$$
By contraction, the first term is $\le \rho\,\|u^\star(f)-u^\star(g)\|_{\mathcal U}$,  by \eqref{eq:one-block}, the second term is $\le \sum_{c=1}^C L_c\|f^{(c)}-g^{(c)}\|_{\infty,2}$.
Therefore,
$\|u^\star(f)-u^\star(g)\|_{\mathcal U}\ \le\ \frac{1}{1-\rho}\,\sum_{c=1}^C L_c\,\|f^{(c)}-g^{(c)}\|_{\infty,2}.$

Finally, we apply the $L_H$-Lipschitzness of $H$ and we obtain:
$$
\|\Phi(f)-\Phi(g)\|_{\mathcal Y}
=\left\|H\left(u^{\star}(f)\right)-H\left(u^{\star}(g)\right)\right\|_{\mathcal Y}
\leq L_H\left\|u^{\star}(f)-u^{\star}(g)\right\|_{\mathcal U}
\leq \frac{L_H}{1-\rho} \sum_{c=1}^C L_c\left\|f^{(c)}-g^{(c)}\right\|_{\infty, 2}.
$$
\end{proof}

\begin{remark}[Interpretation of Proposition \ref{prop:LipPhi}]
Proposition \ref{prop:LipPhi} shows that the final coupled output $\Phi(f)$ is Lipschitz-continuous with respect to the surrogates.
Small uniform perturbations of the surrogates $f^{(c)}$ (measured in $\|\cdot\|_{\infty,2}$) produce controlled changes in the fixed-point solution and thus in the observable output.
The prefactor $\tfrac{L_H}{1-\rho}\sum_c L_c$ can be interpreted as a global "sensitivity" constant of the coupled system with respect to surrogate errors.
\end{remark}

The next corollary applies this general Lipschitz bound to the specific form of perturbations used in the Monte Carlo scheme, namely constant (in $x$) offsets of the GP means at each iteration.

\begin{corollary}[Constant-offset perturbations used in our MC scheme]
\label{cor:offsets}
In our Monte Carlo replication $(m,j)$, we use perturbed surrogates of the form
$\hat f^{(c)}(x)=\bar\mu^{(c)}(x)+\delta^{(c)}_{m,j}$ with $\delta^{(c)}_{m,j}\in\mathbb R^{D_c}$ independent of $x$.
Then $\|\hat f^{(c)}-\bar\mu^{(c)}\|_{\infty,2}=\|\delta^{(c)}_{m,j}\|_2$ and the contraction modulus $\rho$ of $\mathcal T_f$ is unchanged. Hence
\begin{equation}
\big\|\Phi(\hat f)-\Phi(\bar\mu)\big\|_{\mathcal Y}
\ \le\ \frac{L_H}{1-\rho}\,\sum_{c=1}^C L_c\,\|\delta^{(c)}_{m,j}\|_2.
\end{equation}
\end{corollary}

\noindent {Note that $\Phi$ and $\mathcal T_f$ are different.}
$\mathcal T_f:\mathcal U\to\mathcal U$ is the one-step iteration map (depends on $f$ and acts on the interface variable $u$), whereas $\Phi$ maps the surrogate tuple $f$ to the final observable $y=H(u^\star(f))$ obtained at the fixed point of $\mathcal T_f$.

\begin{remark}[Connection to the Monte Carlo scheme]
Corollary \ref{cor:offsets} shows that the specific form of perturbation used in our Monte Carlo scheme
(\emph{constant} offsets with respect to $x$ at each iteration) preserves the contraction property of the coupling operator and turns the Lipschitz bound of Proposition \ref{prop:LipPhi} into a simple bound in terms of the Euclidean norms $\|\delta^{(c)}_{m,j}\|_2$.
This makes it possible to link the GP-based perturbations to a global control on the deviation of the coupled solution.
\end{remark}

Combining the GP posterior variance bounds from the first part of this subsection with the stability result above leads to the following finite-sample, high-probability control on the Monte Carlo outputs, which forms the main theoretical result of this section.

\begin{proposition}[Finite-sample high-probability UQ bound]
\label{prop:finitesampleHQBound_LMC_corrected}
For each $c\in\{1,\dots,C\}$, let the surrogate 
%posterior  
be a multi-output GP with matrix-valued kernel of LMC form
$\kappa^{(c)}(x,x')=\sum_{q=1}^{Q_c} B_q^{(c)}\,\kappa_q^{(c)}(x,x')$
where $B^{(c)}_{q}\succeq \mathbf{0}$, and each scalar latent kernel $\kappa^{(k)}_{q}$ satisfies the scalar fill-distance bound (Theorem \ref{thm:kanagawa_scalar}) with parameters $(s_{c,q},C_{c,q},h_{0,c,q})$, with $s_{c,q}>\tfrac{d_c}{2}$. Let $h^{(c)}$ denote the (local) fill distance of the design used for $f^{(c)}$, and assume
$h^{(c)}\ \le\ h_{0,c}\ :=\ \min_{q} h_{0,c,q}.$ 

In our Monte Carlo scheme, the perturbed surrogates are constant offsets
$$
\hat f^{(c)}(x)=\bar\mu^{(c)}(x)+\delta^{(c)},\qquad \delta^{(c)}\sim\mathcal N(0,\Sigma_c)\ \text{in }\mathbb R^{D_c},\ \ \text{independent of $x$},
$$
with the covariance matrices dominated (in spectral norm) by the posterior marginals: $\lambda_{\max}(\Sigma_c)\ \le\ \sup_{x\in\mathcal X}\lambda_{\max}\!\big(\bar{\kappa}^{(c)}(x,x)\big).$
Then, for any $\beta\in(0,1)$, with probability at least $1-\beta$,
$$
\big\|\Phi(\{\bar\mu^{(c)}+\delta^{(c)}\}_c)-\Phi(\{\bar\mu^{(c)}\}_c)\big\|_{\mathcal{Y}}
\ \le\ \frac{L_H}{1-\rho}\ \sum_{c=1}^C L_c\ 
\sqrt{\,2\,D_c\ \lambda_{\max}\!\Big(\textstyle\sum_{q=1}^{Q_c} B^{(c)}_{q}\Big)\ 
\sum_{q=1}^{Q_c} \mathcal{C}_{c,q}\,\big(h^{(c)}\big)^{\,2s_{c,q}-d_c}\,}\ 
\sqrt{\log\!\Big(\tfrac{2 C D_c}{\beta}\Big)}.
$$

\end{proposition}

\begin{proof}
By Theorem \ref{thm:lmc_fill}, for all $x$,
$\bar{\kappa}^{(c)}(x,x)\ \preceq\ \Big(\sum_{q=1}^{Q_c} \mathcal{C}_{c,q}\,\big(h^{(c)}\big)^{\,2s_{c,q}-d_c}\Big)\ \Big(\sum_{q=1}^{Q_c} B^{(c)}_{q}\Big).$
Taking largest eigenvalues and using homogeneity of $\lambda_{\max}$ for nonnegative scalars gives
$$
\sup_{x\in\mathcal X}\lambda_{\max}\!\big(\bar{\kappa}^{(c)}(x,x)\big)
\ \le\ \lambda_{\max}\!\Big(\sum_{q=1}^{Q_c} B^{(c)}_{q}\Big)\ \sum_{q=1}^{Q_c} \mathcal{C}_{c,q}\,\big(h^{(c)}\big)^{\,2s_{c,q}-d_c}.
$$
By the assumption on $\Sigma_c$, it follows that
$$
\lambda_{\max}(\Sigma_c)\ \le\ \sigma_c^2,
\qquad\text{where}\quad
\sigma_c^2:=\lambda_{\max}\!\Big(\sum_{q=1}^{Q_c} B^{(c)}_{q}\Big)\ \sum_{q=1}^{Q_c} \mathcal{C}_{c,q}\,\big(h^{(c)}\big)^{\,2s_{c,q}-d_c}.
$$

Let $\delta^{(c)}\sim\mathcal N(0,\Sigma_c)$ in $\mathbb R^{D_c}$. For any $t_c>0$, using $\|\cdot\|_2\le \sqrt{D_c}\|\cdot\|_\infty$, a union bound over coordinates, and the one-dimensional Gaussian tail bound,
\begin{align*}
\mathbb P\!\Big(\|\delta^{(c)}\|_2>t_c\Big)
&\le \mathbb P\!\Big(\|\delta^{(c)}\|_\infty>t_c/\sqrt{D_c}\Big) \\
&\le \sum_{\ell=1}^{D_c}\mathbb P\!\Big(|\delta^{(c)}_\ell|>t_c/\sqrt{D_c}\Big)
\ \le\ 2D_c\exp\!\Big(-\frac{t_c^2}{2D_c\,\lambda_{\max}(\Sigma_c)}\Big) \\
&\le\ 2D_c\exp\!\Big(-\frac{t_c^2}{2D_c\,\sigma_c^2}\Big).
\end{align*}
With the choice
$$
t_c:=\sqrt{2D_c}\,\sigma_c\,\sqrt{\log\!\Big(\tfrac{2c D_c}{\beta}\Big)}
=\sqrt{\,2\,D_c\ \lambda_{\max}\!\Big(\textstyle\sum_{q=1}^{Q_c} B^{(c)}_{q}\Big)\ 
\sum_{q=1}^{Q_c} \mathcal{C}_{c,q}\,\big(h^{(c)}\big)^{\,2s_{c,q}-d_c}\,}\ \sqrt{\log\!\Big(\tfrac{2c D_c}{\beta}\Big)},
$$
we get $\mathbb P(\|\delta^{(c)}\|_2>t_c)\le \beta/c$. A union bound over $c=1,\dots,c$ yields that, with probability at least $1-\beta$, simultaneously for all $c$, $\|\delta^{(c)}\|_2\ \le\ t_c.$
Since the offsets are constant in $x$, $\|\hat f^{(c)}-\bar\mu^{(c)}\|_{\infty,2}=\|\delta^{(c)}\|_2$.

By Proposition \ref{prop:LipPhi},
$$
\big\|\Phi(\{\hat f^{(c)}\}_c)-\Phi(\{\bar\mu^{(c)}\}_c)\big\|_{\mathcal Y}
\ \le\ \frac{L_H}{1-\rho}\,\sum_{c=1}^C L_c\,\|\hat f^{(c)}-\bar\mu^{(c)}\|_{\infty,2}
\ \le\ \frac{L_H}{1-\rho}\,\sum_{c=1}^C L_c\, t_c,
$$
and substituting the expression of $t_c$ gives the asserted bounds.
\end{proof}

\begin{remark}[Interpretation and role of Proposition \ref{prop:finitesampleHQBound_LMC_corrected}]
Proposition \ref{prop:finitesampleHQBound_LMC_corrected} is the main theoretical result of this section.
It provides a finite-sample, high-probability bound on the deviation between the output obtained by running the coupling algorithm with the mean surrogates $\{\bar\mu^{(c)}\}_c$
and the random outputs produced by our practical Monte Carlo scheme based on perturbed surrogates $\{\bar\mu^{(c)}+\delta^{(c)}\}_c$. In words, it shows that, with probability at least $1-\beta$, the random outputs $\Phi(\{\bar\mu^{(c)}+\delta^{(c)}\}_c)$ lie in a neighbourhood of the deterministic fixed-point output $\Phi(\{\bar\mu^{(c)}\}_c)$ whose radius is explicitly controlled by: the Lipschitz sensitivity constants of the coupled solver ($L_H$, $\rho$, $L_c$), the posterior variances of the Gaussian process surrogates (through the eigenvalues of $\sum_q B^{(c)}_q$ and the constants $C_{c,q}$), and the fill distances $h^{(c)}$ of the designs used to train the surrogates. As the design becomes denser (so that $h^{(c)}\to 0$ and the posterior variances decrease), the bound shrinks and the Monte Carlo outputs concentrate around the deterministic fixed point.
\end{remark}

\subsection{Analytical validation on a benchmark example}
\label{analyticalExample}

This subsection considers a simple scalar benchmark to compare the surrogate-based coupled solution with a high-accuracy reference
and to assess, in a controlled setting, the behaviour of Method  2 (rigorous trajectory-conditioned sampling) and Method  3
(proposed mean-path constant-offset scheme) introduced in Section  \ref{UQGPMC}.

\medskip
\noindent\textbf{Benchmark and reference solution.}
We consider two deterministic scalar codes (represented by two analytical functions) $g^{(1)},g^{(2)}:[0,1]\to\mathbb{R}$ defined by
\begin{align*}
g^{(1)}(x) &= 0.12 + 0.18x + 0.06\sin(2\pi x)
            + 0.05\exp\!\big(-60(x-0.70)^2\big) + 0.03x(1-x),\\
g^{(2)}(x) &= 0.58 - 0.22x + 0.03\tanh\!\big(8(x-0.40)\big)
            + 0.015\sin(4\pi x).
\end{align*}
The coupled solution $y^\star$ is defined by the fixed-point equation
\begin{equation}
\label{eq:benchmark_fp}
y^\star = \tfrac{1}{2}\big(g^{(1)}(y^\star) + g^{(2)}(y^\star)\big).
\end{equation}

\noindent\textbf{Clarification.}
We embed the interface set in the Banach space $(\mathbb R,|\cdot|)$ and work on the closed subset
$\mathcal U:=[0,1]\subset\mathbb R$. Define the two solvers
 $
\mathcal S_1:\mathcal U\to\mathbb R^2,\quad 
\mathcal S_1(y):=\big(g^{(1)}(y),\,g^{(2)}(y)\big),
$
$
\mathcal S_2:\mathbb R^2\to\mathcal U,\quad 
\mathcal S_2(z_1,z_2):=\tfrac12(z_1+z_2),
$
with transfer operators $\Gamma_1=\mathrm{Id}_{\mathcal U}$, $\Gamma_{12}=\mathrm{Id}_{\mathbb R^2}$ and $\Gamma_2=\mathrm{Id}_{\mathcal U}$.
The coupling operator is therefore
$$
\mathcal T=\Gamma_2\circ\mathcal S_2\circ\Gamma_{12}\circ\mathcal S_1\circ\Gamma_1,
\quad
\mathcal T(y)=\tfrac12\big(g^{(1)}(y)+g^{(2)}(y)\big),
$$
and the coupled solution is the fixed point $y^\star\in\mathrm{Fix}(\mathcal T)$, i.e.\ \eqref{eq:benchmark_fp}.
A numerical bound on $\sup_{y\in[0,1]}|\mathcal T'(y)|$ gives $\rho\simeq 0.279<1$, hence $\mathcal T$ admits a unique fixed point in $[0,1]$.

\medskip
\noindent\textbf{Surrogates and coupled fixed-point solver.}
Training data are generated by Latin hypercube sampling on $[0,1]$ and are noise-free (deterministic observations).
For each code $c\in\{1,2\}$, a scalar GP surrogate $f^{(c)}$ is fitted using a Matérn covariance with smoothness $\nu=5/2$
and a \emph{fixed} length-scale $\ell=0.25$ (no hyperparameter optimisation). A very small nugget effect $\alpha=10^{-12}$ is added
for numerical stability only (it is not interpreted as observation noise).
Figure  \ref{fig:surrogates_overview} shows the resulting surrogates (posterior means and $\pm2\sigma$ bands) for the two DOE sizes
considered below.

For each Monte Carlo replication $j$, the coupled surrogate solution is obtained as the fixed point of
$
y = \tfrac{1}{2}\big(f^{(1,j)}(y)+f^{(2,j)}(y)\big),
$
starting from $y^{(0)}=0.5$ and stopped when $\big|y^{(m+1,j)}-y^{(m,j)}\big|\le\varepsilon$ with tolerance $\varepsilon=10^{-8}$. The two methods differ only in the way the surrogate realizations $f^{(c,j)}$ are generated
inside this coupled iteration (see Section  \ref{UQGPMC} for their detailed definitions).

\begin{figure}[ht!]
    \centering
    \includegraphics[width=0.95\textwidth]{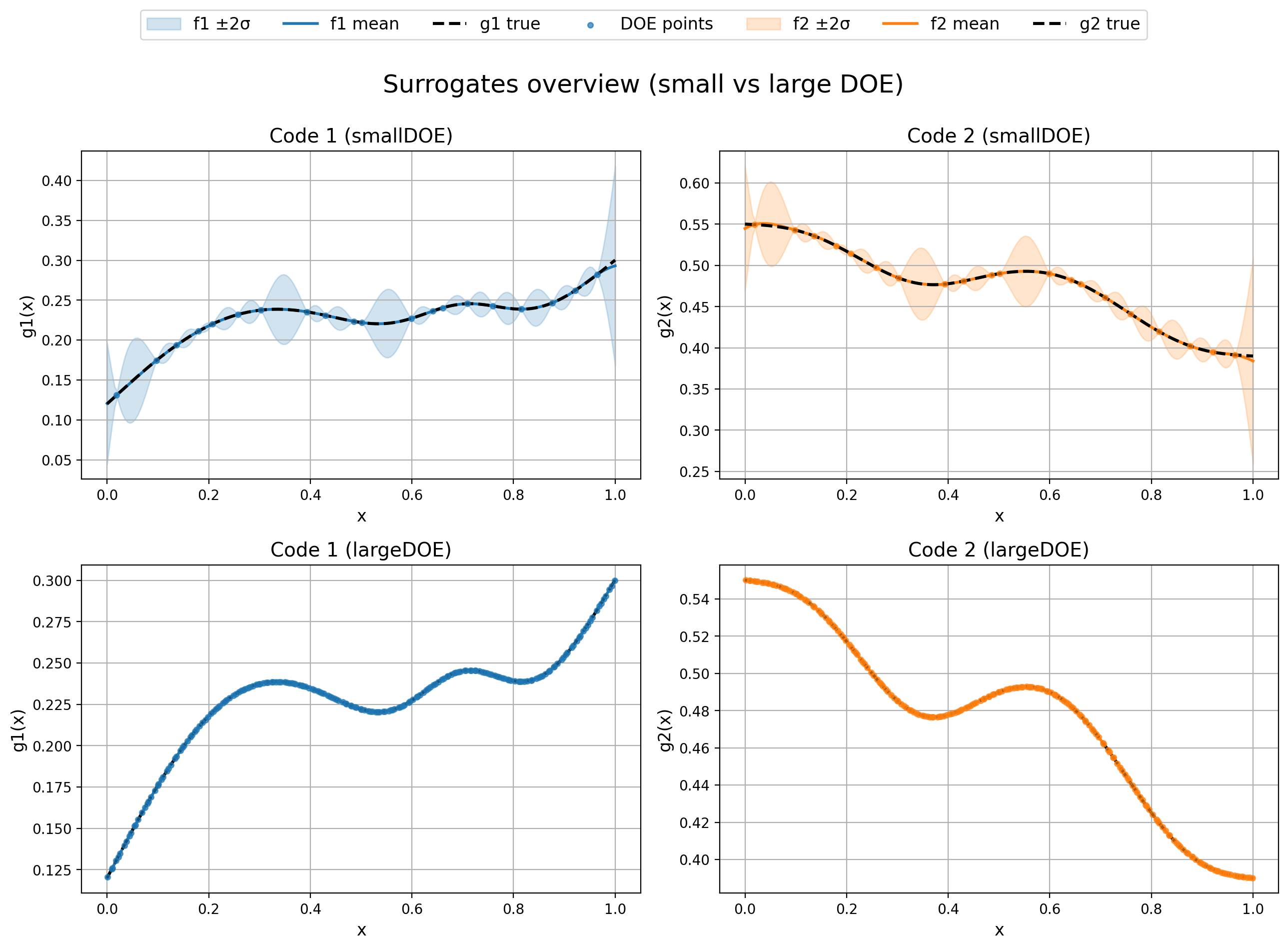}
    \caption{Surrogates for the two scalar codes in the benchmark example, for smallDOE ($n=20$) and largeDOE ($n=200$).
    Deterministic training data, GP means and $\pm2\sigma$ bands are shown together with the true responses.}
    \label{fig:surrogates_overview}
\end{figure}

\noindent\textbf{Benchmark results (small DOE versus large DOE).}
Two DOE sizes are considered: small DOE: $n=20$ training points per code, large DOE: $n=200$ training points per code.
In both cases, the fixed-point tolerance is $\varepsilon=10^{-8}$ and $N=500$ Monte Carlo replications are performed for each method.
The reported $95\%$ intervals are central empirical intervals computed as the $2.5\%$ and $97.5\%$ sample quantiles of the Monte Carlo outputs.

For the smallest design ($n=20$), Method  2 (trajectory-conditioned, rigorous) yields a mean coupled solution of $0.356081$,
with variance $1.623\times 10^{-4}$ and a central $95\%$ empirical interval $[0.330824,\,0.379285]$.
Method  3 (mean-path constant offsets) gives a mean of $0.356523$, variance $2.024\times 10^{-4}$ and a central $95\%$ empirical interval
$[0.327743,\,0.382857]$. Both schemes recover $y^\star=0.3574988$ within these intervals, and the difference in sample means
(M2$-$M3) is $-4.42\times 10^{-4}$. A Welch $t$-test does not detect any significant discrepancy between the means
($t=-0.5177$, $p=0.6048$), while a Kolmogorov--Smirnov test indicates a modest but statistically detectable difference in distribution
shape (statistic $0.0980$, $p=0.0164$).

With a denser design ($n=200$), Method  2 produces a mean of $0.357499$, variance $1.498\times 10^{-10}$ and a central $95\%$ empirical interval
$[0.357475,\,0.357521]$, while Method  3 yields a mean of $0.357500$, variance $2.319\times 10^{-10}$ and a central $95\%$ empirical interval
$[0.357469,\,0.357527]$. The difference in sample means (M2$-$M3) is approximately $-10^{-6}$.
In this regime both distributions are extremely concentrated around $y^\star$, so practical discrepancies are negligible; the Welch $t$-test again
shows no significant mean difference ($t=-1.1371$, $p=0.2558$), while the KS test is borderline at the $5\%$ level
(statistic $0.0860$, $p=0.0495$), suggesting at most a very small difference in distributional shape.

\begin{figure}[ht!]
    \centering
    \includegraphics[width=0.95\textwidth]{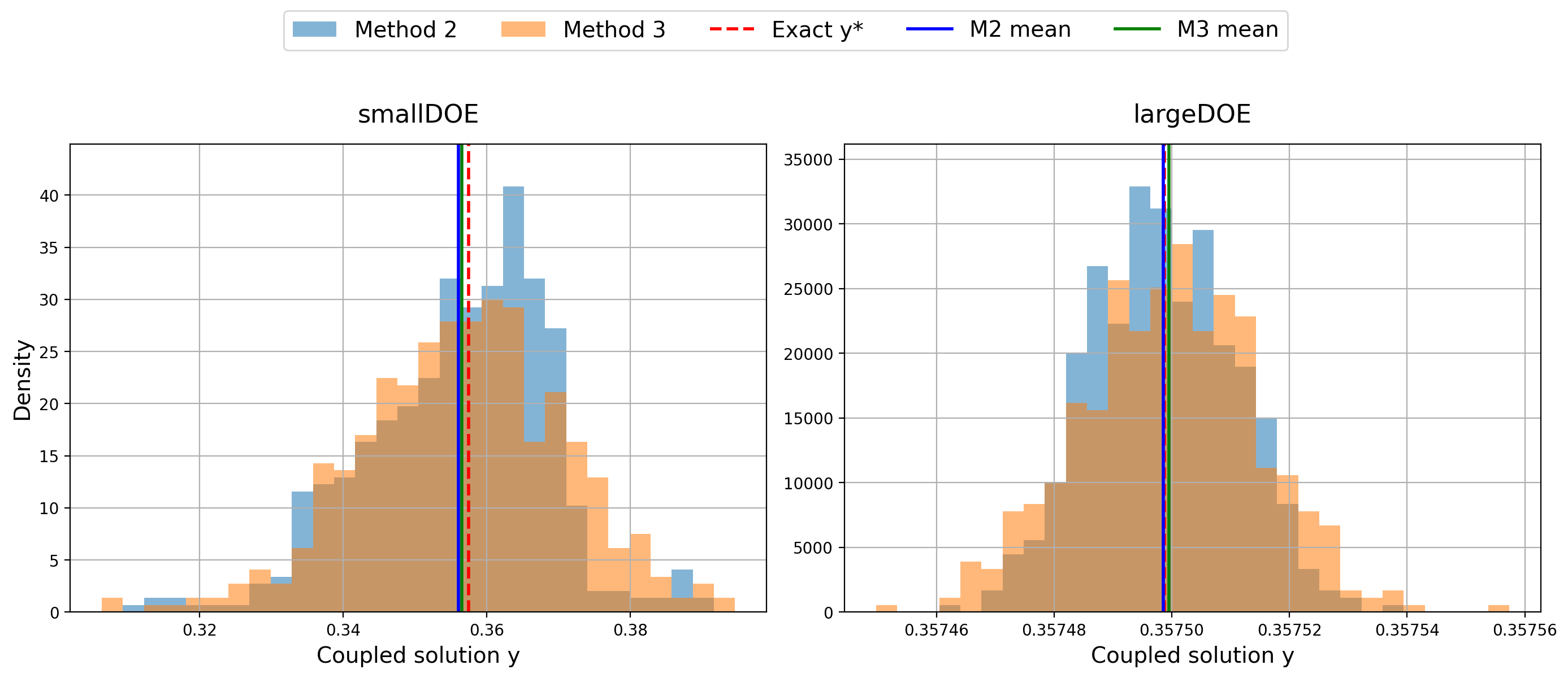}
    \caption{Distribution of the coupled solution $y$ for Method  2 (rigorous trajectory-conditioned sampling) and Method  3 (mean-path constant offsets)
    on the benchmark example, for smallDOE ($n=20$, left) and largeDOE ($n=200$, right). The vertical line indicates the reference solution $y^\star$.}
    \label{fig:benchmark_validation}
\end{figure}

\medskip
\noindent\textbf{Discussion.}
For both DOE sizes, Methods  2 and  3 produce coupled outputs consistent with the reference solution.
For smallDOE, Method  3 matches Method  2 in mean and yields comparable uncertainty levels, while a KS test indicates a modest discrepancy
in distributional shape, for largeDOE, both approaches essentially collapse onto the deterministic fixed point $y^\star$ with variances of order $10^{-10}$.
From a computational standpoint, Method  3 remains significantly cheaper because it reuses a deterministic mean path and requires only a single joint
Gaussian draw on that path, whereas Method  2 performs sequential trajectory conditioning at each iteration. Overall, these numerical observations support
the use of the constant-offset scheme as a controlled and efficient approximation of the rigorous trajectory-conditioned method when the surrogate designs
become sufficiently informative.

%% file: ApplicationToFABow.tex
\subsection{Physical context and coupled model}
\label{sec:phys_context}

Fuel assembly (See Figure \ref{fig:fuelassembly_schematic}) bow is a long-standing safety and operational concern in pressurized water reactor cores, first highlighted in the 1990s through cases of incomplete rod cluster insertion \cite{Andersson2005,kerkar2008exploitation}. The deformation of fuel assemblies results from the combined action of neutronic, mechanical and hydraulic phenomena \cite{karlsson1999modelling,DELAMBERT2019330,WANNINGER2018297}. In industrial practice, fuel assembly bow is evaluated over successive irradiation cycles, each cycle corresponding to about 12-18 months of operation (See Figure \ref{fig:fuelassembly_bow}) and represented in computations by a sequence of quasi-static time steps. At each time step, hydraulic and mechanical conditions are updated and coupled through a fluid-structure interaction (FSI) loop: the hydraulic code Phorcys computes grid-level forces from the current deformed geometry \cite{delam2021}, while the thermomechanical code DACC updates the assembly deformation using a beam-type model of the rod bundle and skeleton \cite{DELAMBERT2023104668}, including nonlinear rod-grid contact, thermal expansion, irradiation growth, creep, variable grid clamping and hold-down forces \cite{abboud2025}. This partitioned Phorcys-DACC coupling is uncertain and can be computationally demanding (slower solving, more iterations) for large systems.

\begin{figure}[ht]
    \centering
    \begin{minipage}[b]{0.48\textwidth}
        \centering
        \includegraphics[height=5cm,width=\textwidth]{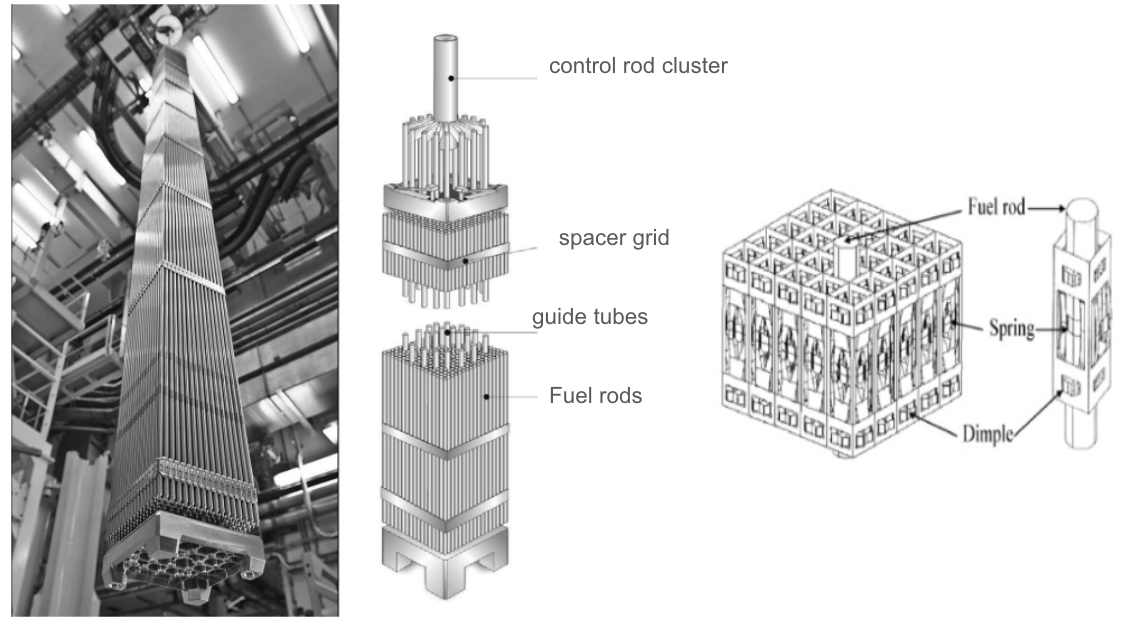}
        \caption{Schematic view of a fuel assembly (around 4m high) with spacer grids \cite{abboudd2025} }
        \label{fig:fuelassembly_schematic}
    \end{minipage}\hfill
    \begin{minipage}[b]{0.48\textwidth}
        \centering
        \includegraphics[height=5cm,width=\textwidth]{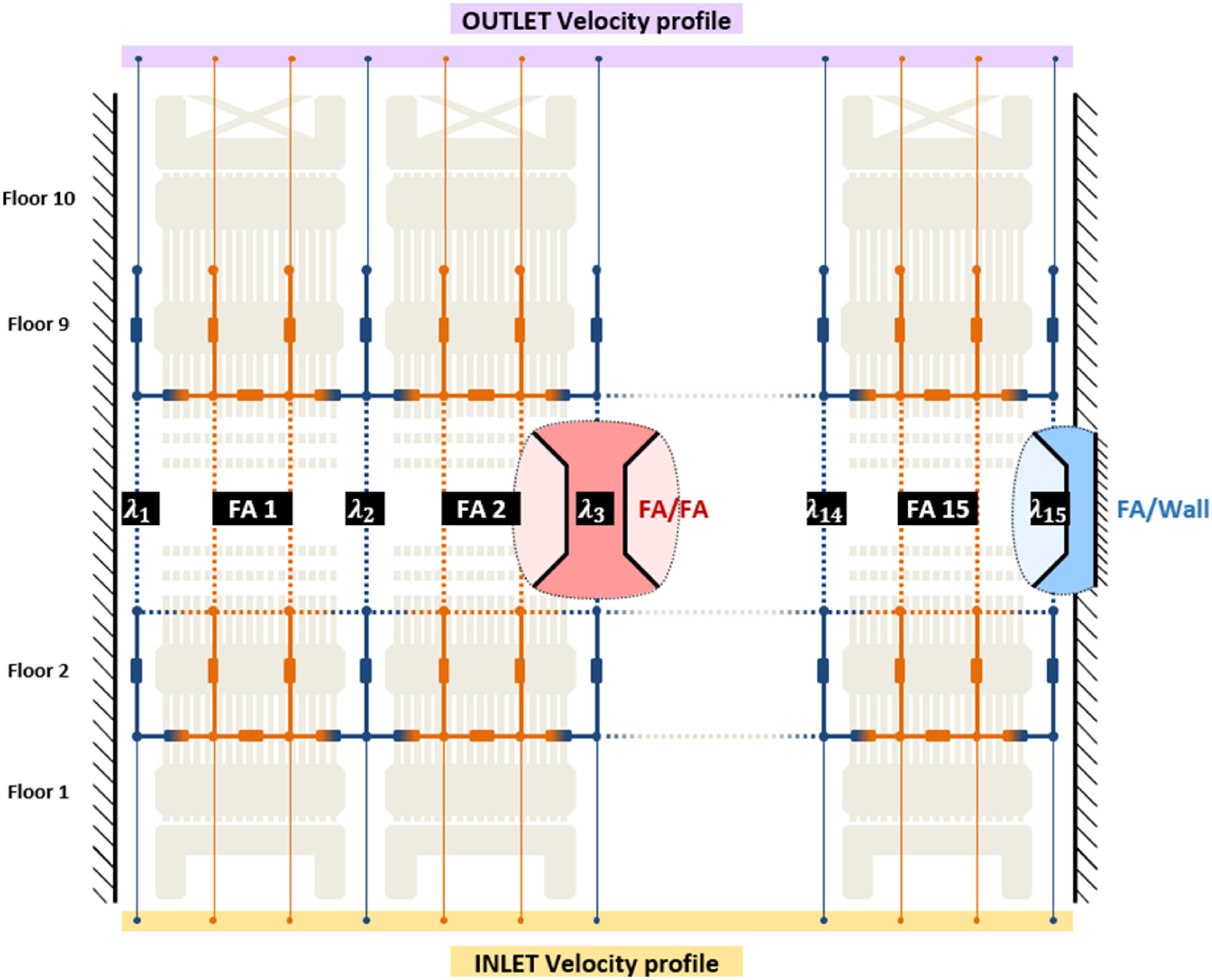}
        \caption{Hydraulic network model for a row of 15 fuel assemblies under reactor operation (inlet and outlet) conditions \cite{delam2021}}
        \label{fig:fuelassembly_bow}
    \end{minipage}
    \label{fig:fuelassembly_overview}
\end{figure}

To reduce this cost while retaining predictive accuracy, surrogate models based on Gaussian processes have been developed in previous works. In \cite{abboud:cea-04598002}, GP hydraulic metamodels were constructed to predict spacer-grid forces from uncertain inputs, and in \cite{abboud2025}, GP mechanical surrogates were built to predict assembly deformation. In both cases, Matérn kernels were adopted to provide flexible but regularised representations of the underlying physics. Their coupling was then formalised and validated in \cite{abboud:cea-05155171}, enabling full-cycle simulations via chained mechanical metamodels, a nonlinear penalisation to prevent inter-assembly penetration, and a robust dual-loop FSI algorithm. Comparisons with the Phorcys-DACC reference chain showed that the coupled surrogates reproduce fuel-assembly bow with sufficient accuracy while reducing computational time by nearly two orders of magnitude.

The present work builds on this surrogate-based coupling and extends it to a fully probabilistic, uncertainty-aware framework. The objective is to quantify the epistemic uncertainty induced by the coupled GP surrogates and to assess, on the realistic fuel-assembly bow application, the practical relevance of the variance-control results established in Section \ref{theoreticalControl}, and finally, to illustrate uncertainties quantification on an example of FA simulation using the surrogate models.

\subsection{Uncertainty quantification of the surrogate models in the context of fuel assembly bow}
\label{sec:UQ_FA}

In this application, the two coupled surrogates play the role of $f^{(1)}$ (hydraulic GP metamodel, one output per spacer grid) and $f^{(2)}$ (mechanical GP metamodel, one output per deformation principal mode). Both are multi-output Gaussian processes with Matérn-based kernels. At each irradiation time step $t$, the coupled Phorcys--DACC solver follows the abstract setting of Section \ref{theoreticalControl} with $C=2$: given a surrogate pair $f=(f^{(1)},f^{(2)})$, the coupling operator $\mathcal T_f$ acts on an interface space $\mathcal U$, admits a unique fixed point $u^\star_t(f)$, and the observable bow field at step $t$ is obtained as
$\mathbf y_t \;=\; \Phi_t(f) \;:=\; H\big(u^\star_t(f)\big),$
with $\Phi_t$ satisfying the Lipschitz stability property of Proposition \ref{prop:LipPhi}. The full-cycle output is obtained by composing these stepwise solution maps over $t=1,\dots,T$. For notational simplicity we denote by $\Phi_{\mathrm{cyc}}(f)$ the resulting full-cycle bow field.  

\noindent\textbf{Deterministic reference coupled simulation.}
Throughout this subsection, we take as deterministic reference the coupled solution obtained by using GP posterior means only:
$f^\mu := \big(\bar\mu^{(1)},\bar\mu^{(2)}\big)$, $\mathbf y^\mu := \Phi_{\mathrm{cyc}}(f^\mu),$
which corresponds to the operational surrogate-based simulator used in practice.

\noindent\textbf{Deterministic mean-path run.}
Uncertainty is propagated with the simplified Monte Carlo strategy (Method 3) of Section \ref{UQGPMC}, adapted to a multi-step irradiation cycle. Starting from the initial core geometry, we first run the coupled simulation with the GP posterior means $\{\bar\mu^{(c)}\}_{c=1}^2$. For each irradiation step $t=1,\dots,T$, the inner fixed-point solver produces interface iterates $(u^{(t,m)})_{m\ge 0}$ and stops after $M_t$ iterations. The surrogate inputs queried at inner iteration $m$ of step $t$ are denoted
$x^{(c)}_{t,m}\in\mathcal X^{(c)},
 t=1,\dots,T, m=1,\dots,M_t, c=1,2.$
We define the deterministic cycle path
$$\mathcal X_{\mathrm{path}}^{(c)}
:= \big\{x^{(c)}_{t,m} : t=1,\dots,T,\ m=1,\dots,M_t\big\}\subset\mathcal X^{(c)},
c=1,2.$$
This set is the cycle-level analogue of the points $\{x^{(c)}_m\}_{m=1}^M$ used in Method 3. It is computed once from the mean-only run and then reused, unchanged, for all Monte Carlo replications.

\noindent\textbf{Monte Carlo perturbations along the deterministic path.}
For each Monte Carlo replication $j=1,\dots,N$ and each surrogate $c\in\{1,2\}$, we sample one joint posterior GP realization restricted to the deterministic cycle path:
$
\big(\tilde f^{(c,j)}(x^{(c)}_{t,m})\big)_{t,m}
 \sim \mathcal N\!\big(\bar\mu^{(c)}(\mathcal X^{(c)}_{\mathrm{path}}),\,\bar{\mathbf{K}}^{(c)} \big),
$
where $\bar\mu^{(c)}(\mathcal X^{(c)}_{\mathrm{path}})=\big(\bar\mu^{(c)}(x^{(k)}_{t,m})\big)_{t,m}$ and $\bar{\mathbf{K}}^{(c)} $ is the block posterior covariance matrix whose $((t,m),(t',m'))$ block is
$\bar\kappa^{(c)} \!\big(x^{(c)}_{t,m},x^{(c)}_{t',m'}\big)\in\mathbb R^{D_c\times D_c}$.
We then define the pathwise offsets
$$
\delta^{(c,j)}_{t,m}
:=
\tilde f^{(c,j)}(x^{(c)}_{t,m})-\bar\mu^{(c)}(x^{(c)}_{t,m}),
\quad t=1,\dots,T,\ \ m=1,\dots,M_t, \ \ j=1,\ldots,N,
$$
and, at irradiation step $t$ and inner iteration $m$, we use the constant-offset perturbed surrogate
$
\hat f^{(c,t,m,j)}(x)
\;=\;
\bar\mu^{(c)}(x)+\delta^{(c,j)}_{t,m}.
$
For fixed $(t,m,j,c)$, the quantity $\delta^{(c,j)}_{t,m}$ is constant in $x$, i.e. it acts as an additive shift of the mean predictor. It is updated only when the algorithm advances to the next inner iteration index $(t,m)$, which matches the constant-offset construction of Method 3 and the setting analysed in Section \ref{theoreticalControl}.

\noindent\textbf{Cycle solve under perturbed surrogates.}
For each replication $j$, we rerun the full-cycle coupled solver from the same initial state as in the mean-path run, replacing $\bar\mu^{(c)}$ by $\hat f^{(c,t,m,j)}$ at each step $t$ and inner iteration $m$. This produces the Monte Carlo output
$\mathbf y^{(j)} \;=\; \Phi_{\mathrm{cyc}}\!\big(\{\hat f^{(c,t,m,j)}\}_{c,t,m}\big).$
The sample $\{\mathbf y^{(j)}\}_{j=1}^N$ then provides a Monte Carlo approximation of the epistemic uncertainty on the coupled cycle output induced by the GP surrogates, in direct correspondence with Method 3. We estimate the mean and covariance of the cycle output by $\hat{\mathbf y}_N=\frac{1}{N}\sum_{j=1}^N \mathbf y^{(j)},
\widehat{\mathrm{Cov}}[\mathbf y]=\frac{1}{N-1}\sum_{j=1}^N\big(\mathbf y^{(j)}-\hat{\mathbf y}_N\big)\big(\mathbf y^{(j)}-\hat{\mathbf y}_N\big)^\top.$
When needed, componentwise summaries (variances, Gaussian confidence intervals) are obtained from $\hat{\mathbf y}_N$ and the diagonal of $\widehat{\mathrm{Cov}}[\mathbf y]$. 

\noindent\textbf{Results and interpretation in light of the theory.}
Figure \ref{USMAB} shows the standard deviation of the coupled GP predictions of fuel-assembly bow over a row of 15 assemblies and across the irradiation cycle, obtained by running the coupled solver with the GP surrogates. The largest epistemic uncertainty induced by the surrogates reaches $0.26\,\mathrm{mm}$. This remains small compared with the displacement scale of about $10\,\mathrm{mm}$. (In the figure, the actual height of the fuel assembly is around $4\,\mathrm{m}$. However, the diagram is intended to represent the spacing between two assemblies, therefore, the vertical scale is not to scale and does not reflect the real length of the fuel assembly. Each small square in the figure has a width equivalent to $2\,\mathrm{mm}$, i.e.\ the nominal water gap between two grids for undeformed fuel assemblies.) This behaviour is consistent with the theoretical results of Section \ref{theoreticalControl}: the training designs for the hydraulic and mechanical surrogates are dense (small fill distances $h^{(c)}$), which keeps posterior variances low along the mean path (Theorem \ref{thm:lmc_fill}). In addition, the Lipschitz stability of the coupled solution map $\Phi$ (Proposition \ref{prop:LipPhi}) limits the amplification of these local surrogate errors through the surrogate-based FSI coupling. Overall, the surrogates capture the fuel-assembly bow response with sufficient accuracy: epistemic uncertainty remains negligible at the coupled-output level, supporting the use of the posterior mean predictors alone in large-scale reactor simulations.

\begin{figure}[!ht]
    \centering
    \includegraphics[width=10cm, height=6cm]{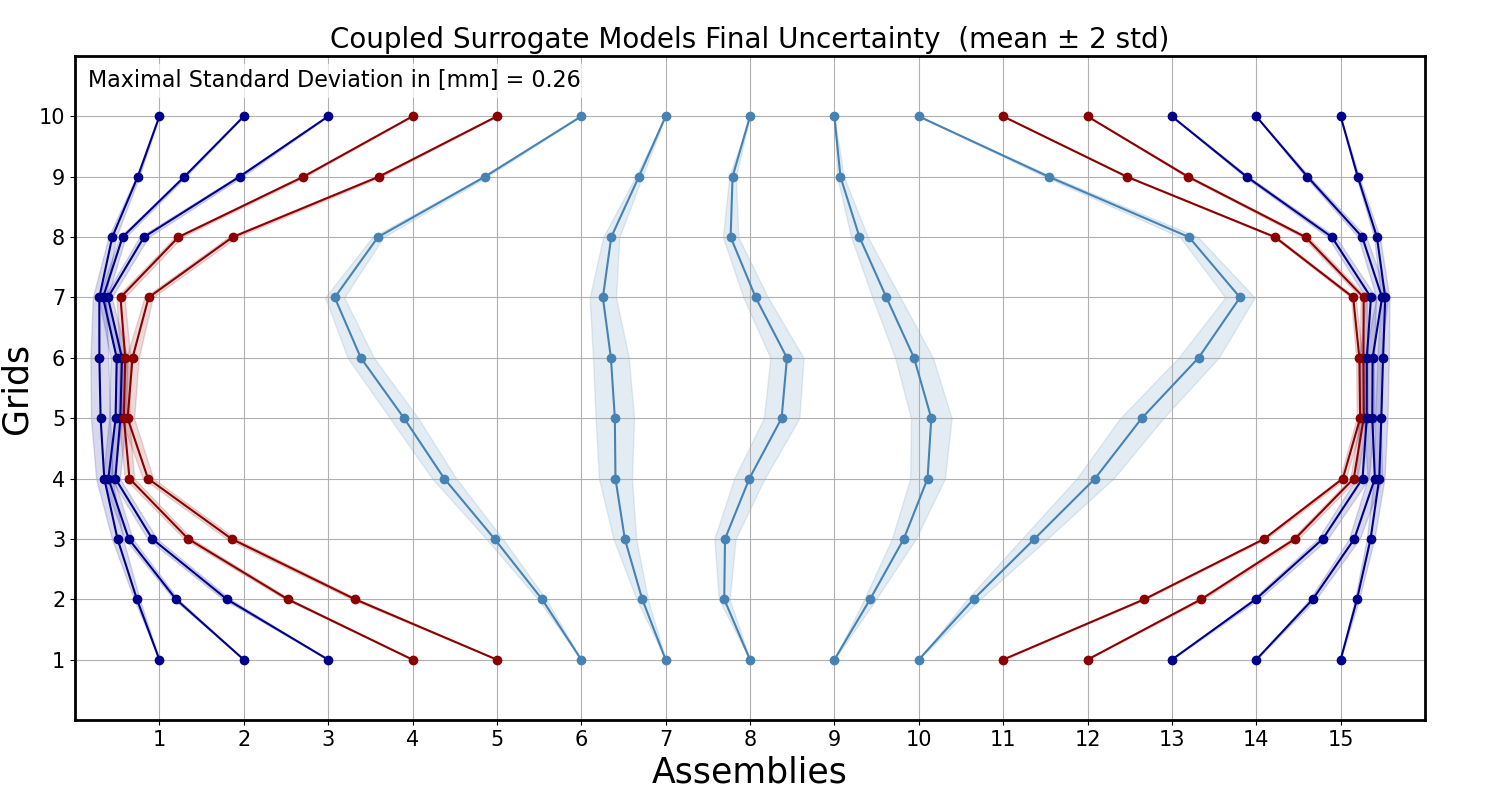}
    \caption{Uncertainty of the coupled GP in a simulation of a row of 15 fuel assemblies}
    \label{USMAB}
\end{figure}

\subsection{Sensitivity analysis for the fluid-structure coupling}
\label{FSIsensitivity}

Based on the results above, we use only the surrogate-model posterior means as predictors in this section and throughout the rest of the paper. We perform a global sensitivity analysis of the coupled fuel-assembly bow simulation using Sobol' indices \cite{Sobol}, leveraging the fast evaluation
and predictive accuracy of the GP metamodels established in \cite{abboud:cea-04598002,abboud2025}.

\medskip
\noindent\textbf{Coupled model with uncertain parameters.}
In addition to the interface variable $u\in\mathcal U$ updated by the fixed-point coupling (Section~\ref{generalcoupling}), the two deterministic solvers depend on a vector of uncertain inputs $\theta\in\Theta\subset\mathbb R^{p}$ (hydraulic boundary conditions, mechanical parameters, and initial deformations). For $c\in\{1,2\}$ we write
$
\mathcal S_c:\ \mathcal X_c\times\Theta\to\mathcal Y_c,\quad (x,\theta)\mapsto \mathcal S_c(x,\theta),
$
and for a fixed $\theta$ we define $\mathcal S_c^\theta(x):=\mathcal S_c(x,\theta)$. For each sample $\theta$ used for Sobol' index estimation, the coupling algorithm is run with $\theta$ kept fixed throughout all fixed-point iterations. The coupled solution therefore depends on $\theta$ through the fixed point
$
u^\star(\theta)\in\mathrm{Fix}(\mathcal T_\theta),
\quad
\mathcal T_\theta := \Gamma_2\circ \mathcal S_2^\theta\circ \Gamma_{12}\circ \mathcal S_1^\theta\circ \Gamma_1,
$
and the quantity of interest is the coupled output $\mathbf y^\star(\theta)=H\big(u^\star(\theta)\big)$.

\medskip
\noindent\textbf{GP surrogates used in the sensitivity study.}
To make the Sobol' analysis tractable, we replace each solver $\mathcal S_c$ by a GP metamodel trained on the joint inputs $(x,\theta)$, and we run the coupling using only its posterior mean $\bar\mu^{(c)}(x,\theta)$ (Section~\ref{UQGPMC}). Accordingly, the sensitivity analysis targets the variability of $\mathbf y^\star(\theta)$ induced by the uncertain inputs $\theta$,
which include both aleatory and epistemic physical parameters through their prescribed probability distributions (See Table \ref{tab:uncertain_params}).

\medskip
\noindent\textbf{High-dimensional input space.}
A direct Sobol' analysis is challenging because $\theta$ is high-dimensional: it includes $45=3\times 15$ modal coefficients describing the three $C$, $S$, and $W$ deformation modes
for fifteen assemblies (see Figure~\ref{fig:CSW}), five hydraulic boundary-condition parameters, and additional uncertain quantities from the hydraulic and mechanical models.

\begin{figure}[ht!]
    \centering
    \includegraphics[width=7cm, height=6cm]{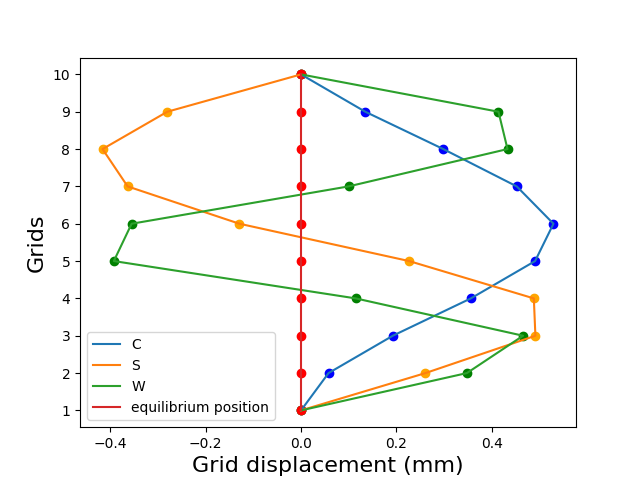}
    \caption{C,S,W deformation modes of a fuel assembly}
    \label{fig:CSW}
\end{figure}

\noindent
\textbf{Methodology.} To make this analysis computationally feasible and scientifically meaningful, a structured reduction of the input space is applied. This reduction is guided by earlier, sensitivity analyses carried out separately for the mechanical and hydraulic models in \cite{abboud:cea-04598002,abboud2025}. The input parameters are organized into coherent groups to reflect the underlying physical processes: Initial modal deformations are aggregated into three representative vectors for the $C$, $S$, and $W$ modes, hydraulic boundary conditions are grouped into a distinct set of corresponding parameters, the results of the sensitivity analyses on both the hydraulic and mechanical models allowed us to identify and neglect the variability of certain parameters that exhibited negligible influence on the simulation outputs. These include the inlet temperature in both simulations, the axial resistance of spacer grids in the hydraulic model, and the fast neutron flux. 
It is important to note that these parameters were not entirely removed from the models but rather fixed at their mean values, because these parameters originally served as input variables for the metamodels. This structured reformulation reduces redundancy, simplifies interpretation, and enhances the clarity of the resulting Sobol' indices.

\noindent
\textbf{The choice of the initial configuration and the uncertain input parameters.} Based on a prescribed neutronic configuration, each fuel assembly is positioned within the core according to its irradiation cycle (an assembly of cycle $j$ has already experienced $j-1$ full irradiation cycles). Assemblies of cycle 3 are placed in the central region, assemblies of cycle 2 are placed at the periphery, and assemblies of cycle 1 are placed in intermediate positions (see Figure \ref{initialuncertainty}). This influences the input parameters of the mechanical metamodels, especially the modal coefficients which are used to generate the initial deformations as a function of the age of each assembly. Additionally, the grid clamping force is determined individually for each assembly. Table \ref{tab:uncertain_params} summarizes the uncertain input parameters used in the coupled simulations, along with their associated probability distributions. {All input parameters in this study are supposed to be independent.}

\begin{table}
    \centering
    \begin{tabular}{|c|l|c|}
        \hline
        \textbf{Symbol} & \hspace*{0.3in} \textbf{Name} & \textbf{Probability Distribution} \\
        \hline
        $C_{j}$ & Modal coefficient C & $\mathcal{N}(0, \sigma^{2}_C)$ ($\sigma_C$ depends on irradiation cycle) \\\hline
        $S_{j}$ & Modal coefficient S & $\mathcal{N}(0, \sigma^{2}_S)$ ($\sigma_S$ depends on irradiation cycle) \\\hline
        $W_{j}$ & Modal coefficient W & $\mathcal{N}(0, \sigma^{2}_W)$ ($\sigma_W$ depends on irradiation cycle) \\\hline
        \multicolumn{3}{|c|}{\textbf{Hydraulic boundary conditions}} \\\hline
        $V$ & Average axial velocity & $\mathcal{N}(\mu = 5,\sigma^{2} = 0.05^2)$  \\ \hline
        $M_{in}$ & Maximum deviation (inlet) & $\mathcal{N}(\mu = 0.05, \sigma^{2} = 0.005^2)$  \\ \hline
        $L_{offset}^{in}$ & Inlet lateral offset & $\mathcal{N}(\mu = 0, \sigma^{2} = 0.2^2)$  \\ \hline
        $M_{out}$ & Maximum deviation (outlet) & $\mathcal{N}(\mu = 0.04,\sigma^{2} = 0.004^2)$  \\ \hline
        $L_{offset}^{out}$ & Outlet lateral offset & $\mathcal{N}(\mu = 0, \sigma^{2} = 0.1^2)$  \\ \hline
        \multicolumn{3}{|c|}{\textbf{Epistemic parameter}} \\\hline
        $h_l$ & Epistemic parameter (hydraulic model) & $\mathcal{U}(0.01, 0.03)$   \\\hline
        \multicolumn{3}{|c|}{\textbf{Mechanical parameters}} \\\hline
        $MSI$ & Holddown interaction of the system & $\mathcal{U}(0.01, 0.025)$ \\\hline
        $grid\_clamping$ & Grid clamping force & $\mathcal{N}(\mu, \sigma^{2})$ ($\mu,\sigma^{2}$ depend on irradiation cycle)\\\hline
        $C_{creep}$ & Creep constant & $\mathcal{N}(\mu = 1, \sigma^2 = 0.3^2)$ \\\hline
        $C_{growth}$ & Growth constant & $\mathcal{N}(\mu = 1, \sigma^2 = 0.3^2)$ \\
        \hline
    \end{tabular}
    \vspace{0.2cm}
    \caption{Uncertain input parameters, their names, and associated probability distributions.}
    \label{tab:uncertain_params}
\end{table}

\noindent
\textbf{Quantities of Interest and Sensitivity Analysis.} The primary quantities of interest in this study are the modal coefficients \(C\), \(S\), and \(W\), which characterize the mechanical deformation state of each fuel assembly. A global sensitivity analysis has been conducted using Sobol' indices, which quantify the contribution of each uncertain input parameter to the variance of these outputs. In this study, the number of simulations is set to $n_{\text{simulation}} = 11{,}000$, which corresponds to the number required to compute the Sobol' sensitivity indices using Saltelli method \cite{Saltelli2000} implemented in the \texttt{Uranie} platform \cite{blanchard:cea-02052632}. The Saltelli approach requires a total of $n_s (n_x + 2)$ simulations, where $n_s = 1{,}000$ denotes the base sample size, and $n_x = 9$ represents the number of input variables: three vector inputs ($C$, $S$, and $W$), one vector for the hydraulic boundary conditions, and the five scalar inputs. The analysis has been performed at three key stages of the reactor cycle: prior to irradiation, at the end of the irradiation phase, and after irradiation during the opening of the upper core plate (See Figure \ref{CycleStep}).
   \begin{figure}[ht]
    \centering
         \includegraphics[width=14cm, height=7.cm]{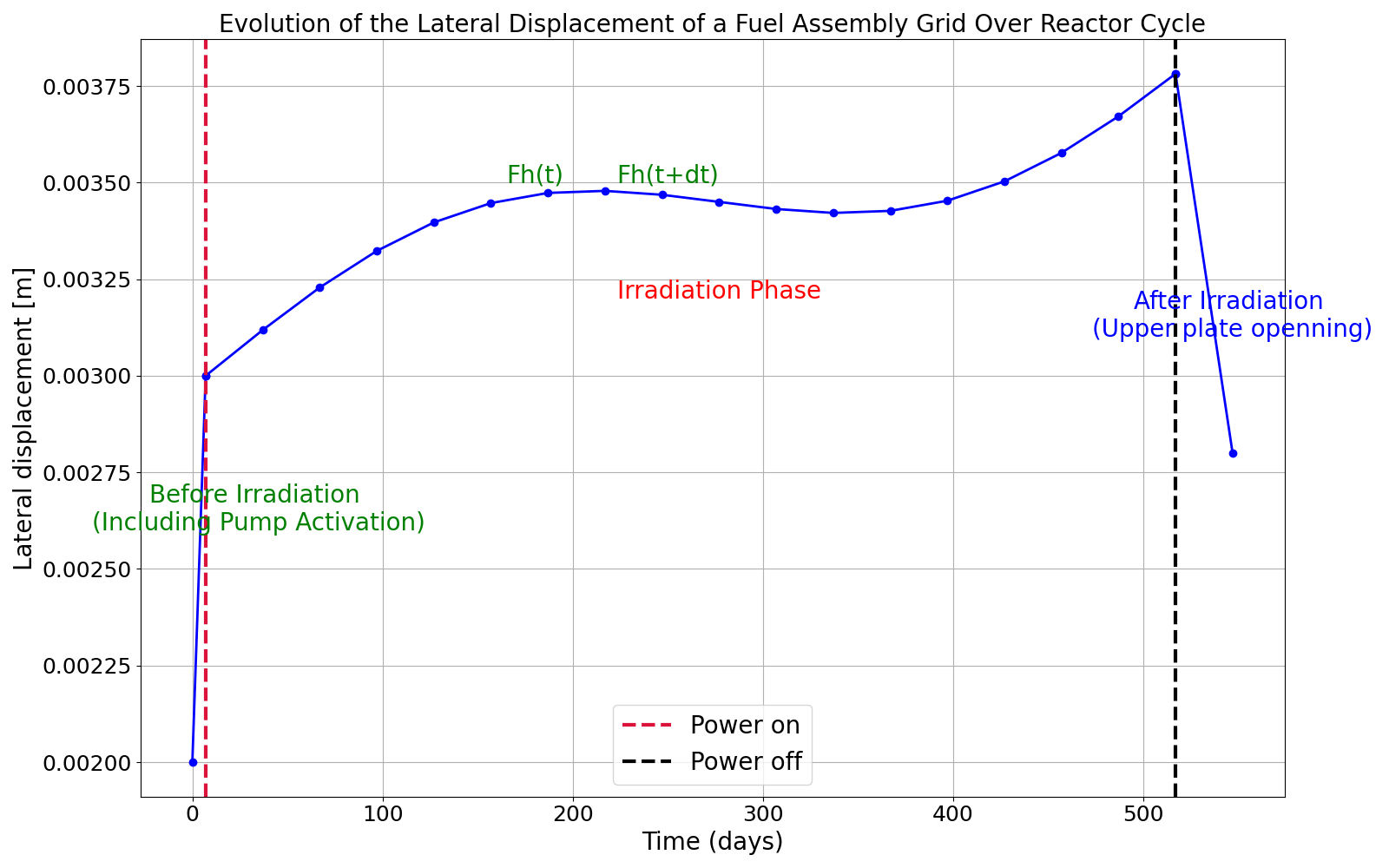}
        \caption{Schematic of the fuel assembly grid’s lateral displacement across reactor cycle stages.}
        \label{CycleStep}
        \end{figure}

This temporal segmentation allows the evolution of sensitivities to be captured throughout the operational cycle. In addition, a detailed Sobol' analysis has been conducted by grouping fuel assemblies according to their irradiation age (cycle) and spatial location in the core. 

Old assemblies of cycle 3 are located at the center of the core,
assemblies of cycle 2 are placed near the core shroud,
and new assemblies of cycle 1 are positioned in the intermediate regions.
This classification enables a more precise interpretation of how age and radial position influence mechanical behavior. To further enhance interpretability, we have computed aggregated Sobol' indices \cite{iooss}, which quantify the relative contribution of groups of input parameters to the total output variance. For clarity of presentation, these aggregated indices are displayed as normalized pie charts, thereby providing a compact visual summary of the dominant sources of uncertainty across assemblies and time points. It is worth emphasizing that the sum of the aggregated first-order Sobol' indices has been found to be close to unity, indicating that higher-order interactions among input parameters contribute negligibly to the variance. This property justifies both the normalization step for visualization and the use of aggregated first-order indices as a reliable measure of parameter importance in this context.

\noindent
\textbf{Results before irradiation phase.} 
During the pre-irradiation phase, the simulation considers the insertion force of control rods in the grids, closure of the core, activation of primary circulation pumps, and the thermal ramp-up to operational temperature. This phase corresponds to the initial mechanical stabilization of the reactor core under thermal and hydraulic loads, in the absence of irradiation-driven phenomena such as creep or growth.
%
%\noindent
%\textbf{Key Observations:} 
Figure \ref{aggsobolbefore} shows that the initial geometric deformation of fuel assemblies is the dominant factor influencing the final deformation state across all cycles (1, 2, and 3) before irradiation. These initial deformations determine the contact distribution between assemblies and thus govern the mechanical equilibrium at the early stages.

   \begin{figure}
    \centering
         \includegraphics[width=10cm, height=8cm]{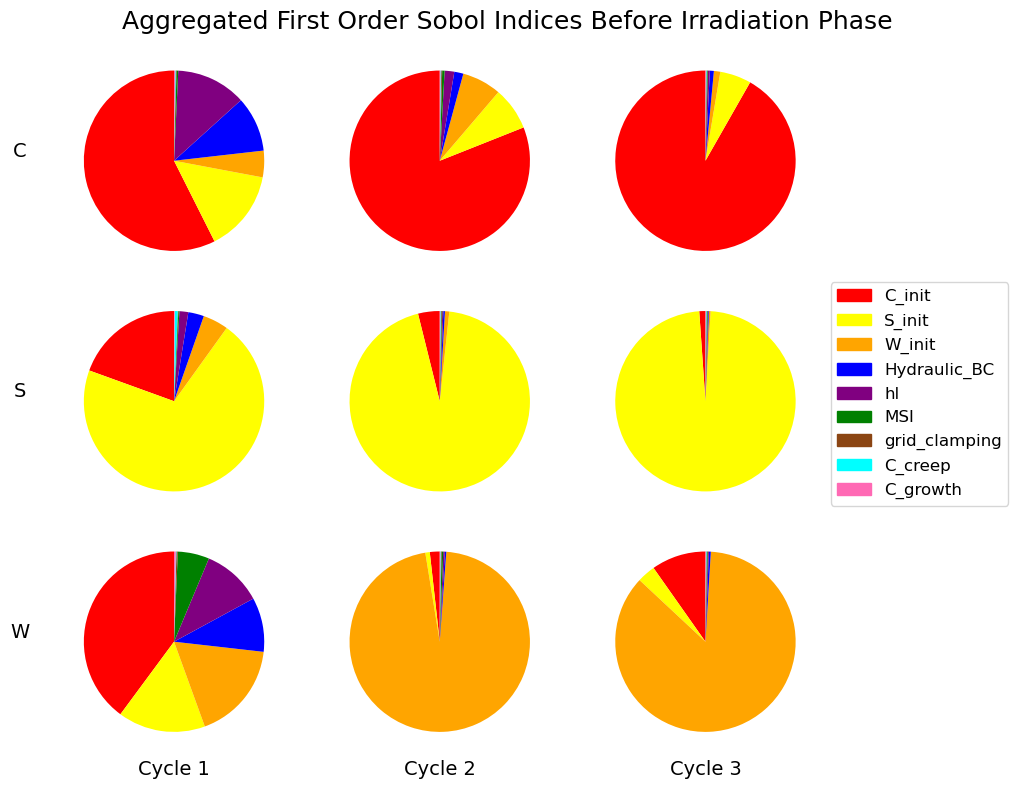}
        \caption{Aggregated first order Sobol' indices before irradiation phase on the C,S,W deformation modes}
        \label{aggsobolbefore}
        \end{figure}

\noindent
\textbf{Results at the end of irradiation phase.} 
At the end of irradiation, deformation is governed by both irradiation-driven effects and the legacy of initial conditions (Figure \ref{aggsobolirrad}). Overall, hydraulic boundary conditions and the epistemic parameter hl seem to dominate the variance across cycles, especially regarding modes C and W. For assemblies of {Cycle 1}, where initial deformations are small, the response is mainly controlled by BC and \(h_l\), with creep acting as a secondary contributor, especially for mode \(S\). For assemblies of Cycle 2, BC and $h_l$ remain dominant for C and W modes, and creep remains a notable second-order effect; while regarding S mode, larger initial deformations play a stronger role, influence is shared between initial deformations, BC, $h_l$, and creep. For assemblies of {Cycle 3}, initial deformations become the leading driver in several modes (notably \(S\)), surpassing the hydraulic influence, while creep has only marginal impact. Growth- and hold-down-related parameters remain negligible throughout.

   \begin{figure}
    \centering
         \includegraphics[width=10cm, height=8cm]{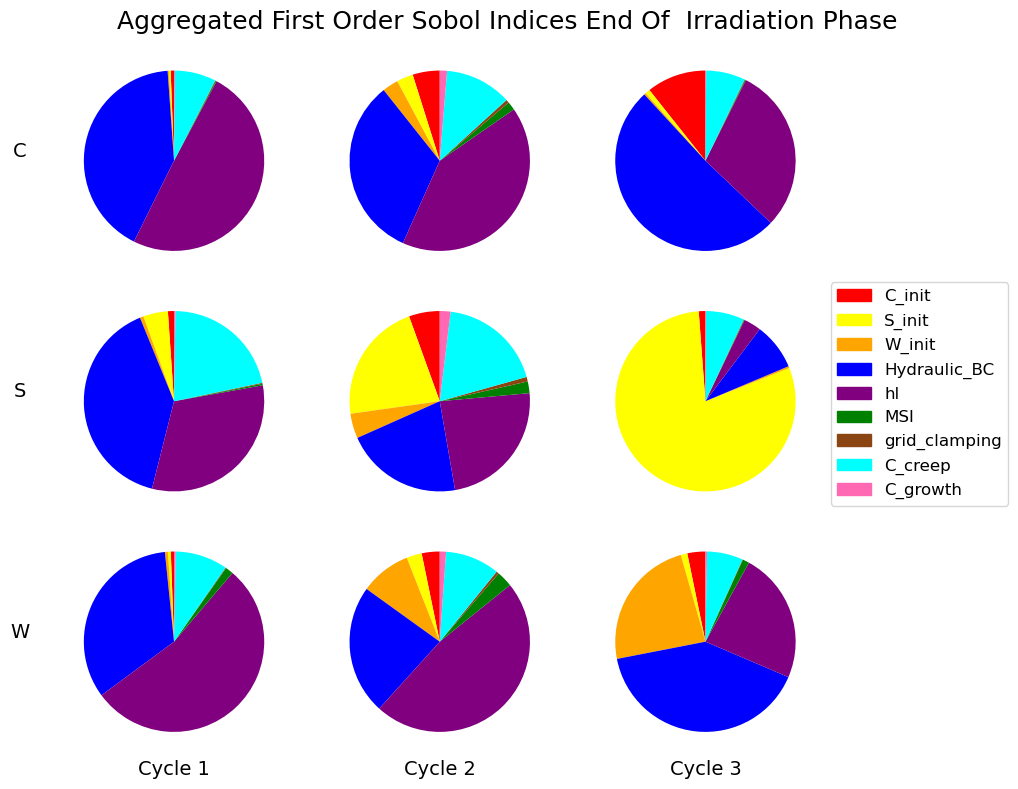}
         \caption{Aggregated first order Sobol' indices at the end of irradiation phase on the C,S,W deformation modes}
        \label{aggsobolirrad}
        \end{figure}

 \noindent
 \textbf{Results after irradiation phase.} 
 At the post-irradiation stage -typically following upper core plate lift-off- assemblies undergo elastic relaxation and retain residual deformation shaped by their irradiation history and boundary conditions (Figure \ref{aggsobolafter}). For assemblies of {Cycle 1} with low burnup in the intermediate regions, modal deformation variance is governed mainly by hydraulic boundary conditions (BC) and the epistemic parameter $h_l$, with creep contributing only second-order effects.  For assemblies of Cycle 2, regarding modes S and W, on can argue initial deformations, BC and $h_l$ remain influent, however no input really stands out in mode $C$.  For assemblies of {Cycle 3} at the core center with highest burnup, while BC and $h_l$ remain influent regarding  C and W modes (along with $C$ and $W$ initial deformations), the S mode is almost only driven by legacy geometric deformations (initial $S$ mode), confirming the persistence of accumulated deformations. Overall, results for Cycle 3 are similar to the ones in the end-of-irradiation phase, highlighting that phenomena occurring after this phase only have little impact on input for these assemblies. In this regime, creep becomes negligible compared to geometric effects.

   \begin{figure}
    \centering
         \includegraphics[width=10cm, height=8cm]{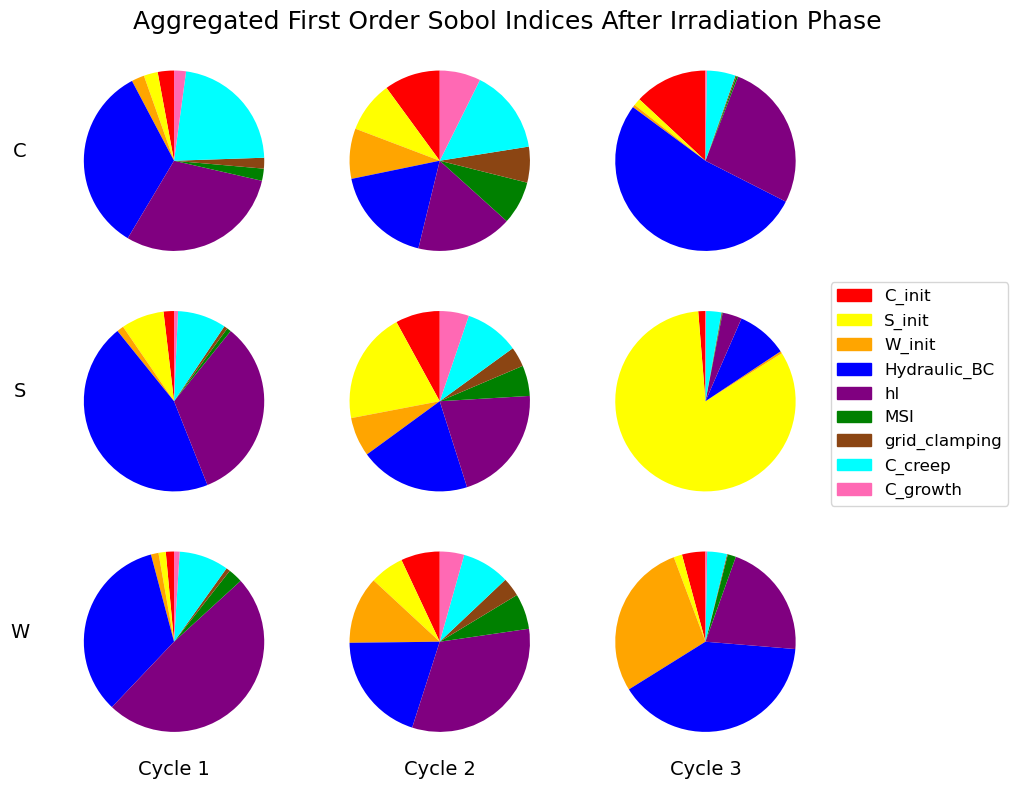}
         \caption{Aggregated first order Sobol' indices after the irradiation phase on the C,S,W deformation modes}
        \label{aggsobolafter}
        \end{figure}

\subsection{Predictive Accuracy of the Simulation Codes}

\label{accuracy}

In simulations involving coupled mechanical and hydraulic models of fuel assemblies, initial deformation data are subject to measurement uncertainty. Though these uncertainties are often overlooked, it has been demonstrated that they play a significant role in the variance of the final outcomes \cite{abboud2025}. It is therefore crucial to investigate how such uncertainties propagate through the computational model and influence the resulting physical quantities. This section primarily aims to evaluate the predictive capability of the coupled model for a given configuration-referred to as the 'calculation of the day'-by quantifying the confidence that can be placed in its results, assuming the input measurements (including initial deformations) are known up to a certain measurement error. A secondary objective is to analyze the contribution of initial deformation uncertainty to the overall result variability, in order to better understand the dominant sources of uncertainty.

\noindent
\textbf{Methodology.} 
Running high-fidelity coupled simulations that account for both the mechanical and hydraulic behavior of nuclear fuel assemblies can become time-consuming for large systems, especially when such simulations must be repeated numerous times for uncertainty quantification. To overcome this limitation, we use the surrogate models coupling strategy. These tools offer an efficient and reasonably accurate alternative to direct numerical simulations, significantly reducing computation time while preserving the essential characteristics of the physical response. It is important to note that, in this simulation, we use the mean values of the Gaussian processes as predictors, and we do not account for the uncertainty associated with the surrogate model itself. The approach adopted in this study relies on a combination of physical assumptions and data inspired by realistic reactor operating conditions.

\medskip

\noindent\textbf{Hydraulic Modeling.} 
The parabolic velocity profile introduced in \cite{abboud:cea-05155171} is assumed. The uncertainty in these  boundary conditions is considered to be small. The hydraulic epistemic parameter, denoted $h_l$, is treated as a uniformly distributed variable within the range $[10\,\text{mm},\,30\,\text{mm}]$, with a nominal reference value of $20\,\text{mm}$.

\medskip

\noindent\textbf{Geometrical Deformation.} 
The initial deformation is applied to a row of 15 fuel assemblies, as described in Section \ref{FSIsensitivity}). In this study, we do not explore the full space of possible deformation patterns; rather, we focus specifically on the measurement uncertainty surrounding the known (measured) bow shapes. To this end, random variations are introduced in the modal coefficients $C$, $S$, and $W$ to simulate the expected variability due to measurement errors.

\medskip

\noindent\textbf{Material and Structural Uncertainties.} 
Beyond geometric variability, the model incorporates uncertainties in several material and structural properties. Parameters such as creep behavior, fuel rod growth, and the stiffness of grid spacers-particularly those related to clamping forces-are included in the analysis, as their uncertainties cannot be neglected. The probability distributions used for these parameters are those listed in Table \ref{tab:uncertain_params}.

\noindent
\textbf{Propagation of measurement uncertainty on modal coefficients.}\label{sec:modal_uncertainty}
In the fuel-assembly bow application, the initial deformation is measured at 10 axial levels (typically at the spacer grids) and represented in a reduced modal basis. Let $\mathbf{U}\in\mathbb{R}^{10}$ denote the true displacement vector and write $\mathbf{U} \;=\; \sum_{j=1}^3 C_j\,\mathbf{u}_j \;=\; \mathbf{M}\mathbf{C},$
where $\{\mathbf{u}_1,\mathbf{u}_2,\mathbf{u}_3\}$ is an orthonormal family of deformation modes, $\mathbf{M}\in\mathbb{R}^{10\times 3}$ collects these modes as columns and $\mathbf{C}=(C_1,C_2,C_3)^\top$ are the modal coefficients. The measured displacement is modeled as
$\mathbf{U}_{\mathrm{mes}} \;=\; \mathbf{U} + \boldsymbol{\varepsilon},
\quad
\boldsymbol{\varepsilon}\sim\mathcal{N}(\mathbf{0},\sigma^2\mathbf{I}_{10}),
\quad
\sigma = 0.3 \text{mm}.$ Since the modes are orthonormal, the least-squares estimate of $\mathbf{C}$ is simply the projection of $\mathbf{U}_{\mathrm{mes}}$ onto this basis. The projection is linear and the measurement noise is Gaussian, so the estimated modal coefficients are themselves Gaussian, mutually uncorrelated, with variance $\sigma^2$ inherited from the measurement noise. This justifies modeling the uncertainties on $C_j$, $S_j$ and $W_j$ as independent Gaussian variables with variance $\sigma^2$.

\medskip
\noindent
\textbf{Representation of the uncertainty in the initial configuration.}
An assumed measurement uncertainty of approximately $0.3$ mm is therefore applied on the modal coefficients to represent variability in the initial deformation of fuel assemblies. Figure \ref{initialuncertainty} illustrates the resulting uncertainty in the initial deformation field of the reference simulation, showing the $95\%$ confidence intervals associated with each assembly. In this study, the same uncertainty level ($0.3$ mm) is applied regardless of the fuel cycle stage. This is likely conservative at the beginning of cycle 1: factory measurements, performed in air on controlled metrology benches with unirradiated assemblies, generally exhibit smaller uncertainties, and assemblies can be straightened if necessary. In contrast, measurements taken later in the cycle (after cycles 1-3) are made under water on irradiated assemblies and are typically less precise. Throughout this section, the nominal water-gap thickness (distance between two grids, corresponding to one square in Figure \ref{initialuncertainty}) is taken as $2$ mm at full power.

\noindent
\textbf{Uncertainty on the parabolic velocity profile.} 
The axial velocity distribution is represented by a quadratic law $v(x) = a x^2 + b x + c$ where $x$ is the spatial coordinate in the core and the coefficients $(a,b,c)$ are uncertain. They are inferred from three Gaussian random inputs: the mean axial velocity $\bar{v} \sim \mathcal{N}(\mu_v,\sigma_v^2)$, the maximum deviation $M \sim \mathcal{N}(\mu_M,\sigma_M^2)$, and the lateral offset $L_{off} \sim \mathcal{N}(\mu_{L_{off}},\sigma_{L_{off}}^2)$. 
Algebraic constraints link these inputs to $(a,b,c)$, defining the parabolic shape. Uncertainty propagation is performed by Monte Carlo simulation: generate $N$ realizations of $(\bar{v},M,L_{off})$, compute coefficients $(a_i,b_i,c_i)$ and evaluate $v_i(x)=a_i x^2+b_i x+c_i$, estimate the empirical mean and variance at each $x$,
$\hat{v}(x) = \frac{1}{N}\sum_{i=1}^N v_i(x), 
    \quad
    \sigma_{v(x)}^2 =  \frac{1}{N-1}\sum_{i=1}^N \big(v_i(x)-\hat{v}(x)\big)^2$.
    %(iv) construct the pointwise $95\%$ confidence interval.
%
This procedure quantifies the variability of $v(x)$ induced by uncertainty in the input parameters. Figure \ref{velocityprofileuncertainty} illustrates the resulting dispersion of velocity profiles consistent with the prescribed probability laws.

\begin{figure}[h!]
    \centering
    \begin{minipage}[t]{0.48\textwidth}
        \centering
        \includegraphics[width=\textwidth,height=6cm]{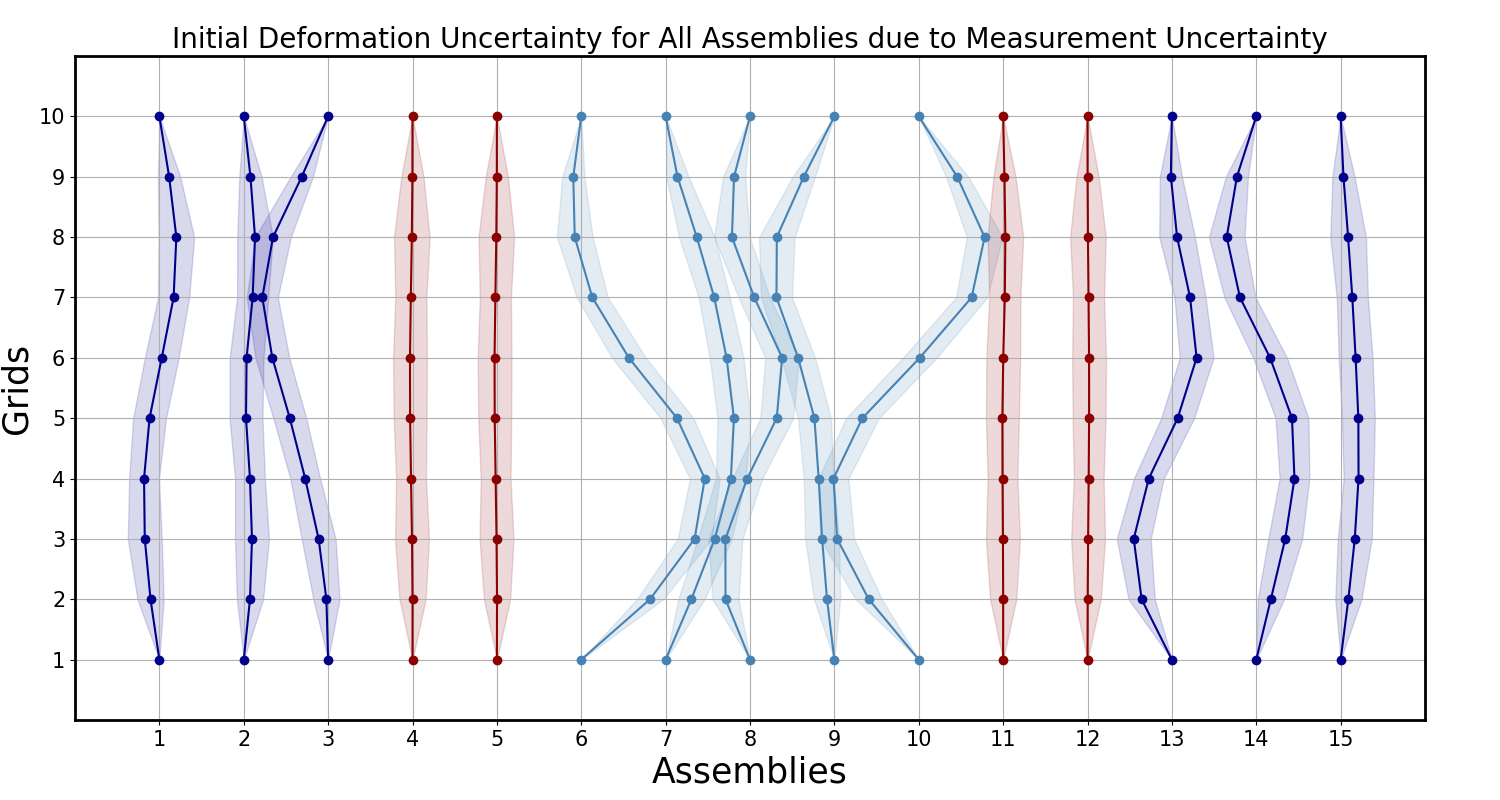}
        \captionof{figure}{Uncertainty in the initial deformation field with the 95\% confidence intervals associated with each assembly (dark blue: assemblies of cycle 2; red: assemblies of cycle 1; light blue: assemblies of cycle 3)}
        \label{initialuncertainty}
    \end{minipage}%
    \hfill
    \begin{minipage}[t]{0.48\textwidth}
        \centering
        \includegraphics[width=\textwidth,height=6cm]{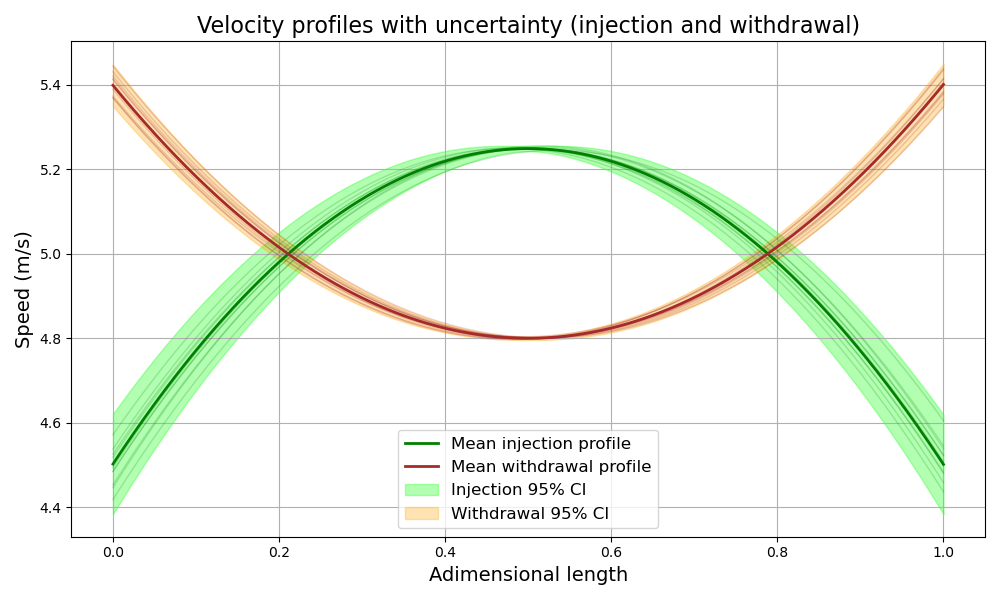}
        \captionof{figure}{Uncertainty in the boundary condition profiles with 95\% confidence interval (injection profile = inlet conditions and withdrawal profile = outlet conditions in Figure \ref{fig:fuelassembly_bow}).}
        \label{velocityprofileuncertainty}
    \end{minipage}
\end{figure}

\noindent
\textbf{Reference simulation and uncertainty quantification.} 
The reference simulation corresponds to a fully coupled thermo-mechanical and hydraulic computation performed under nominal conditions. In this context, all uncertain input parameters are fixed at their respective mean values. The geometric configuration considered is that of a row of 15 fuel assemblies. The simulation covers a complete irradiation period of 18 months, during which the evolution of mechanical deformation is tracked. The primary quantity of interest is the final deformation of each assembly in the row.

To assess the impact of parameter uncertainties on the simulation output, a forward uncertainty quantification procedure is conducted using a Monte Carlo approach. A total of $2{,}000$ simulations are performed, each corresponding to a distinct realization of the input parameters sampled from their associated probability density function. These uncertain inputs include the uncertainty of the measurements of the initial geometry (the measured value, assumed to be read from the instrument, has been fixed), material behavior parameters (e.g., creep and growth coefficients), and the hydraulic parameters. The empirical variance and confidence intervals of the output deformation fields are then estimated from the simulation results.

In parallel, a modal decomposition is employed to represent the deformation field \( Y(z) \) along the axial coordinate \( z \). This representation is expressed as: $Y(z) = \sum_{i=1}^3 C_i u_i(z),$
where \( \{u_i(z)\}_{i=1}^3 \) are deterministic modal shapes, and \( \{C_i\}_{i=1}^3 \) are the corresponding random modal coefficients. These coefficients encapsulate the uncertainties arising from measurement errors of the initial geometry. 
%By employing the previously established mechanical modes. 
Assuming the modal shapes are fixed and known, the variance of the deformation field \( Y(z) \) at each axial location can be computed using the following expression:
$\text{Var}(Y(z)) = \sum_{i,j=1}^3 u_{i}(z)u_{j}(z)\, \text{Cov}(C_i, C_j),$ 
where \( \text{Cov}(C_i, C_j) \) denotes the covariance between modal coefficients \( C_i \) and \( C_j \), which is estimated from the sampled data set. This formulation allows the construction of a spatial uncertainty map of the deformation field. Moreover, it facilitates the derivation of local 95\% confidence intervals at each position \( z \).
%\begin{equation}
   %  \left[ \mathbb{E}[Y(z)] - 1.96\, \sqrt{\text{Var}(Y(z))},\ \mathbb{E}[Y(z)] + %1.96\, \sqrt{\text{Var}(Y(z))} \right].
%\end{equation}
Such analysis provides a detailed and probabilistically rigorous characterization of the possible variations in the final deformation profiles, enhancing the interpretability and reliability of the simulation results under uncertainty.

\noindent
\textbf{Results.} 
Figure \ref{refSimulation} shows the reference final deformation after 18 months of irradiation and removal of the upper plate.
The results are now presented in three figures to facilitate analysis by grouping them according to the type of fuel assembly and their spatial location within the row. 
Figure \ref{cycle2} corresponds to assemblies of cycle 2 positioned at the periphery of the reactor vessel. Figure \ref{cycle1} shows the results for assemblies of cycle 1 located in the intermediate region, and Figure \ref{cycle3} presents the results for assemblies of cycle 3 placed at the center of the core. In each figure, the standard deviation of the grid displacement is computed. This provides an indicator of the maximum uncertainty associated with each type of assembly across the three figures.

\begin{figure}[h!]
    \centering
    % Left minipage: reference + cycle 2
    \begin{minipage}[t]{0.48\textwidth}
        \centering
        \includegraphics[width=\textwidth,height=5cm]{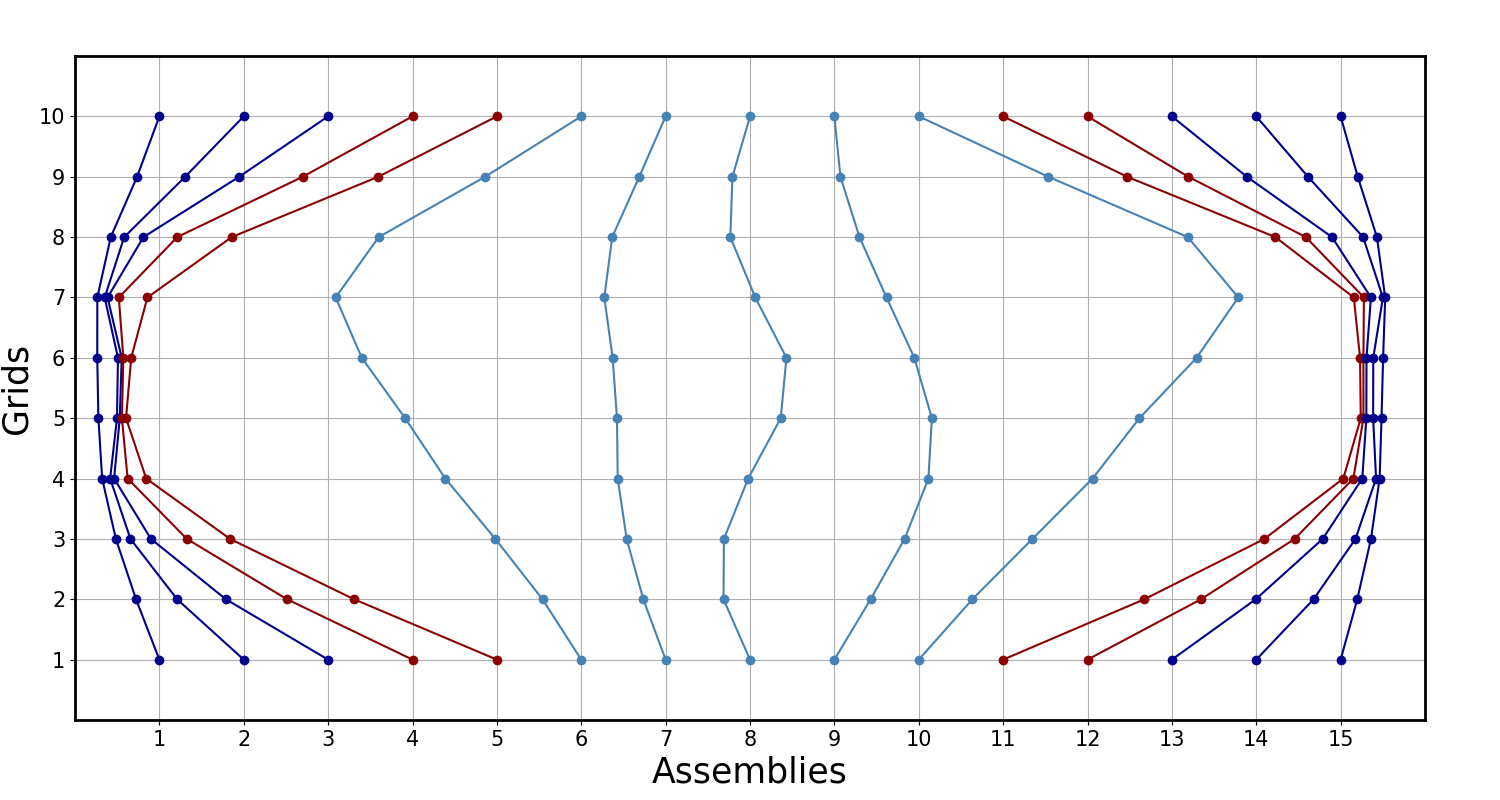}
        \captionof{figure}{Reference simulation results after an 18-months irradiation cycle}
        \label{refSimulation}
        
        \vspace{0.5em} % small vertical spacing
        \includegraphics[width=\textwidth,height=5cm]{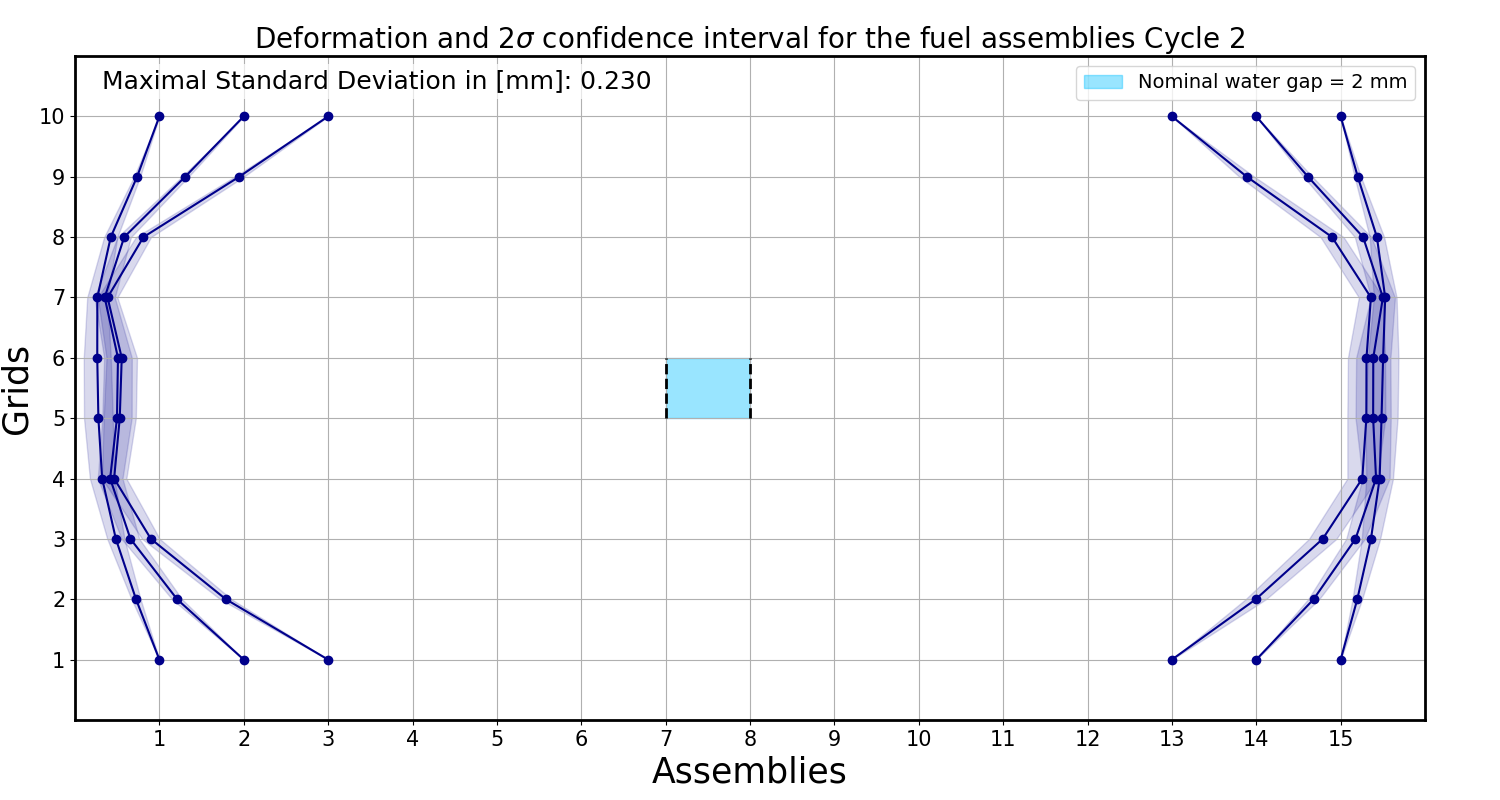}
        \captionof{figure}{Uncertainty in the final deformation of fuel assemblies of cycle 2}
        \label{cycle2}
    \end{minipage}%
    \hfill
    % Right minipage: cycle 1 + cycle 3
    \begin{minipage}[t]{0.48\textwidth}
        \centering
        \includegraphics[width=\textwidth,height=5cm]{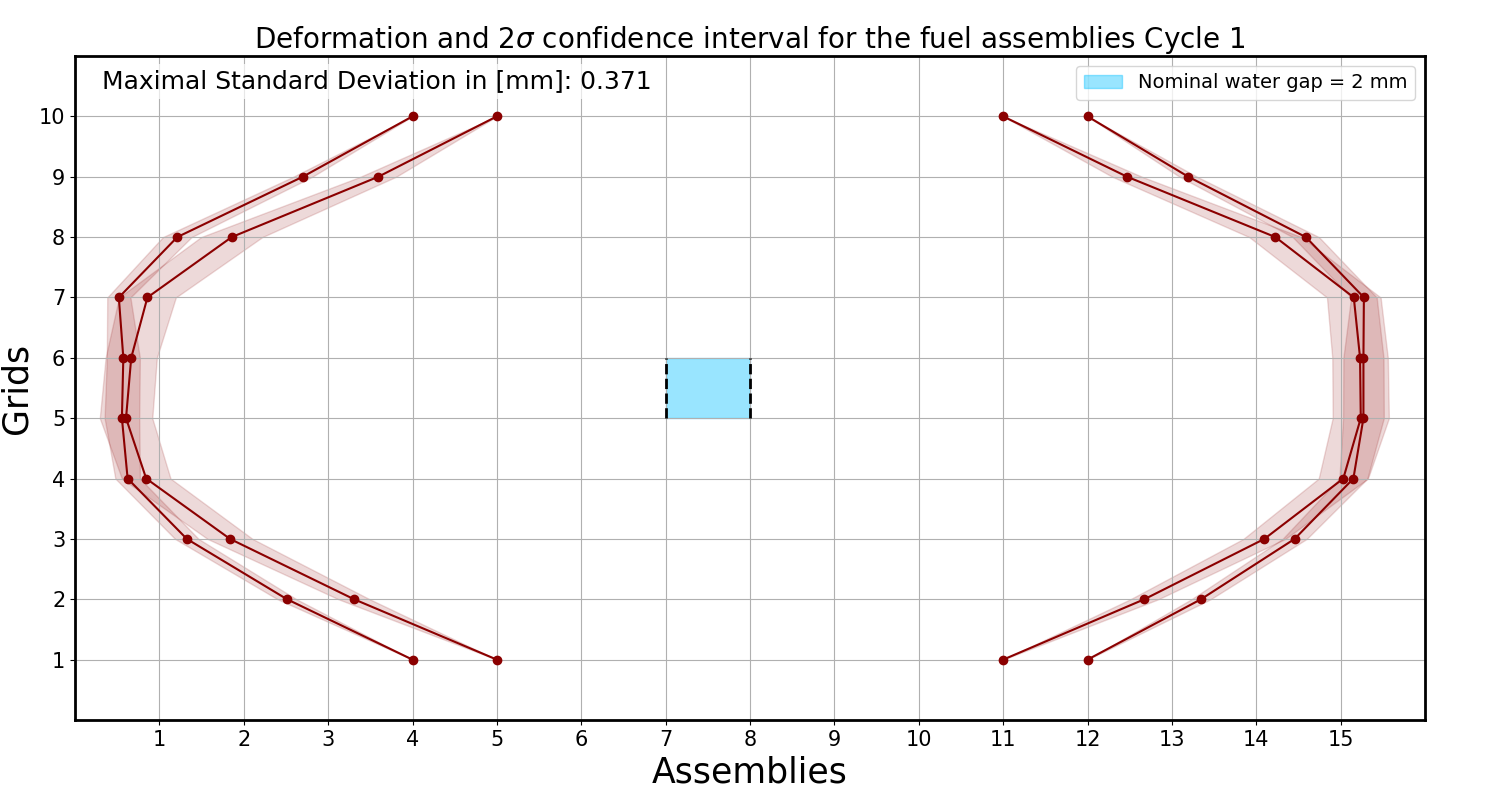}
        \captionof{figure}{Uncertainty in the final deformation of fuel assemblies of cycle 1}
        \label{cycle1}
        
        \vspace{0.5em} % small vertical spacing
        \includegraphics[width=\textwidth,height=5cm]{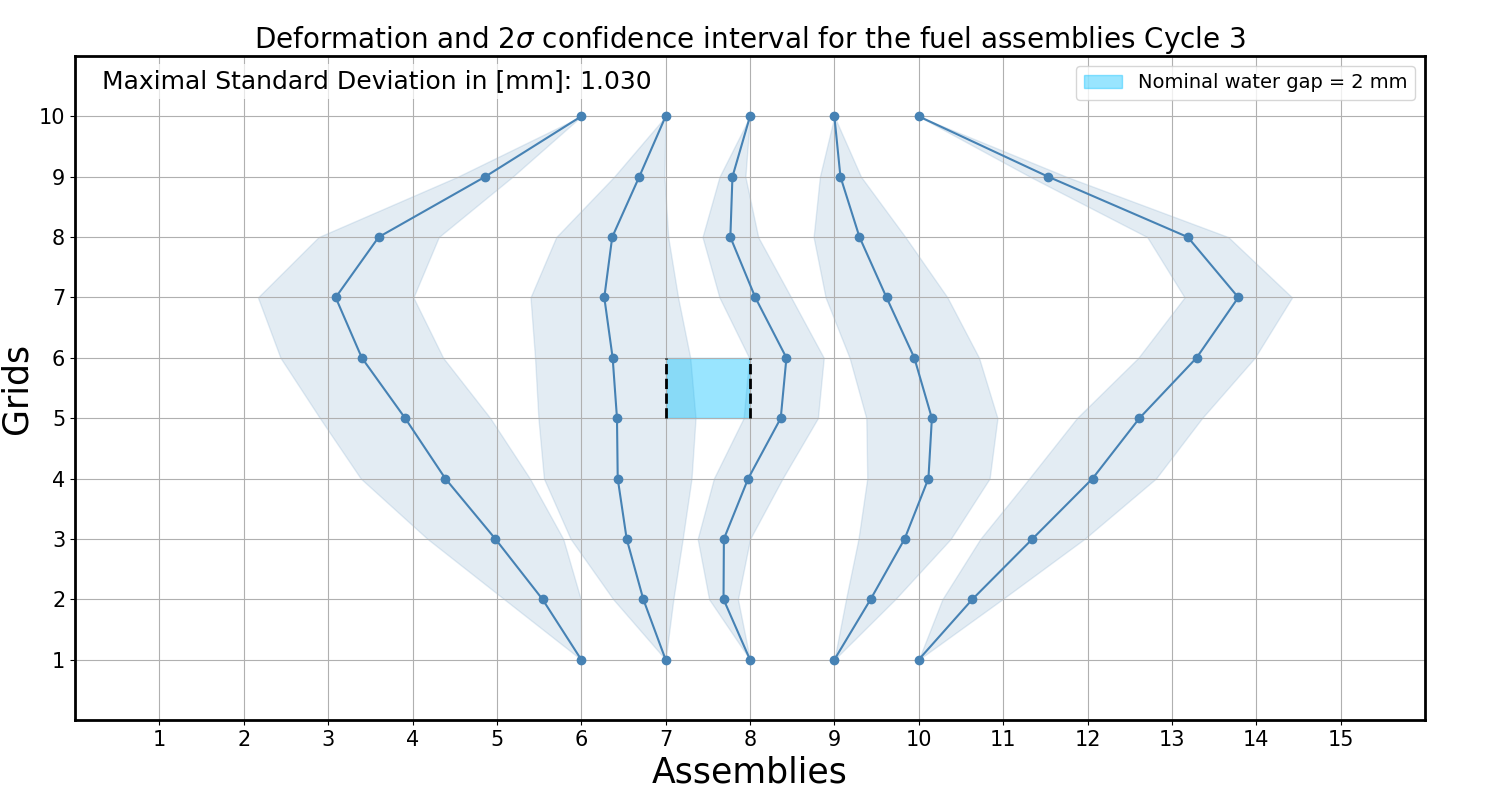}
        \captionof{figure}{Uncertainty in the final deformation of fuel assemblies of cycle 3}
        \label{cycle3}
    \end{minipage}
\end{figure}

After presenting the uncertainty estimates in the previous figures, we now show the full distributions of $2{,}000$ calculations for assemblies $A_6$ and $A_{10}$ in Figure \ref{deformationsall}, using the sample mean as a summary statistic. The distributions clearly deviate from  Gaussian, which can be attributed to nonlinear effects and interactions between complex physical phenomena.

%\noindent \textbf{Interpretation}
The uncertainty in the simulation results seems primarily influenced by two key factors: the spatial position of the assemblies within the core and their irradiation age. Our interpretation is that assemblies located near the  {core center} are more free to deform due to fewer mechanical constraints of their neighbors, which leads to higher variability in their displacement responses. In contrast,  {peripheral assemblies} are restricted by the reactor vessel, limiting their movement and consequently reducing the associated uncertainty. Additionally, the {age of the assemblies} plays a significant role. Older assemblies deform more under the same loading, which in turn increases the uncertainty linked to their mechanical behavior. With increasing irradiation cycles, fuel assemblies undergo stress relaxation of the tiny springs retaining the rods and, as a consequence, tend to loose stiffness. As a result, older assemblies are more prone to bowing than newly loaded, straight assemblies under identical loading conditions. These two parameters-position and age-must therefore be carefully considered when interpreting results given by the coupled codes.

\begin{figure}[h!]
    \centering
    \begin{minipage}[t]{0.45\textwidth}
        \centering
        \includegraphics[width=0.7\textwidth,height=6cm]{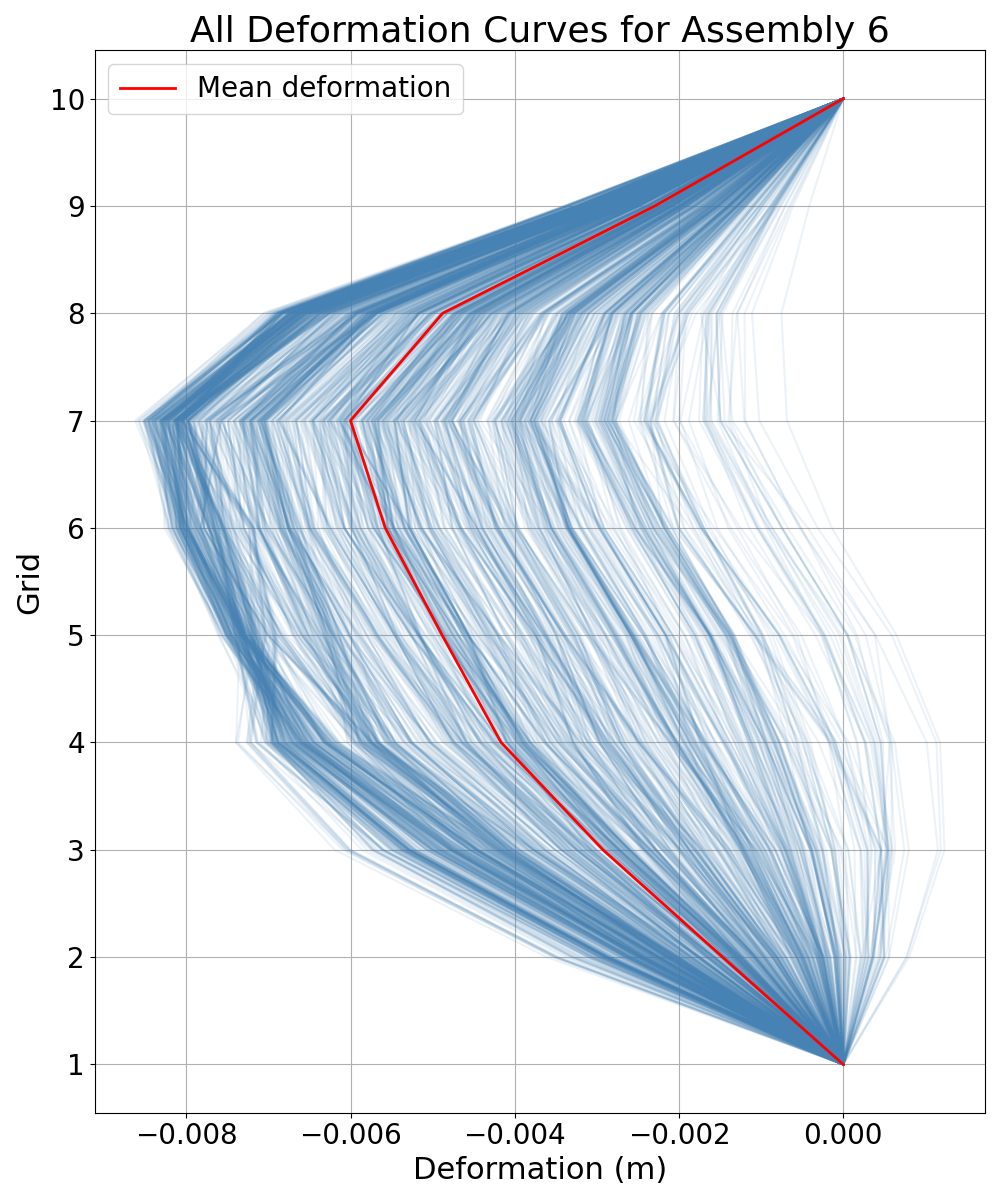}
    \end{minipage}%
    \hfill
    \begin{minipage}[t]{0.45\textwidth}
        \centering
        \includegraphics[width=0.7\textwidth,height=6cm]{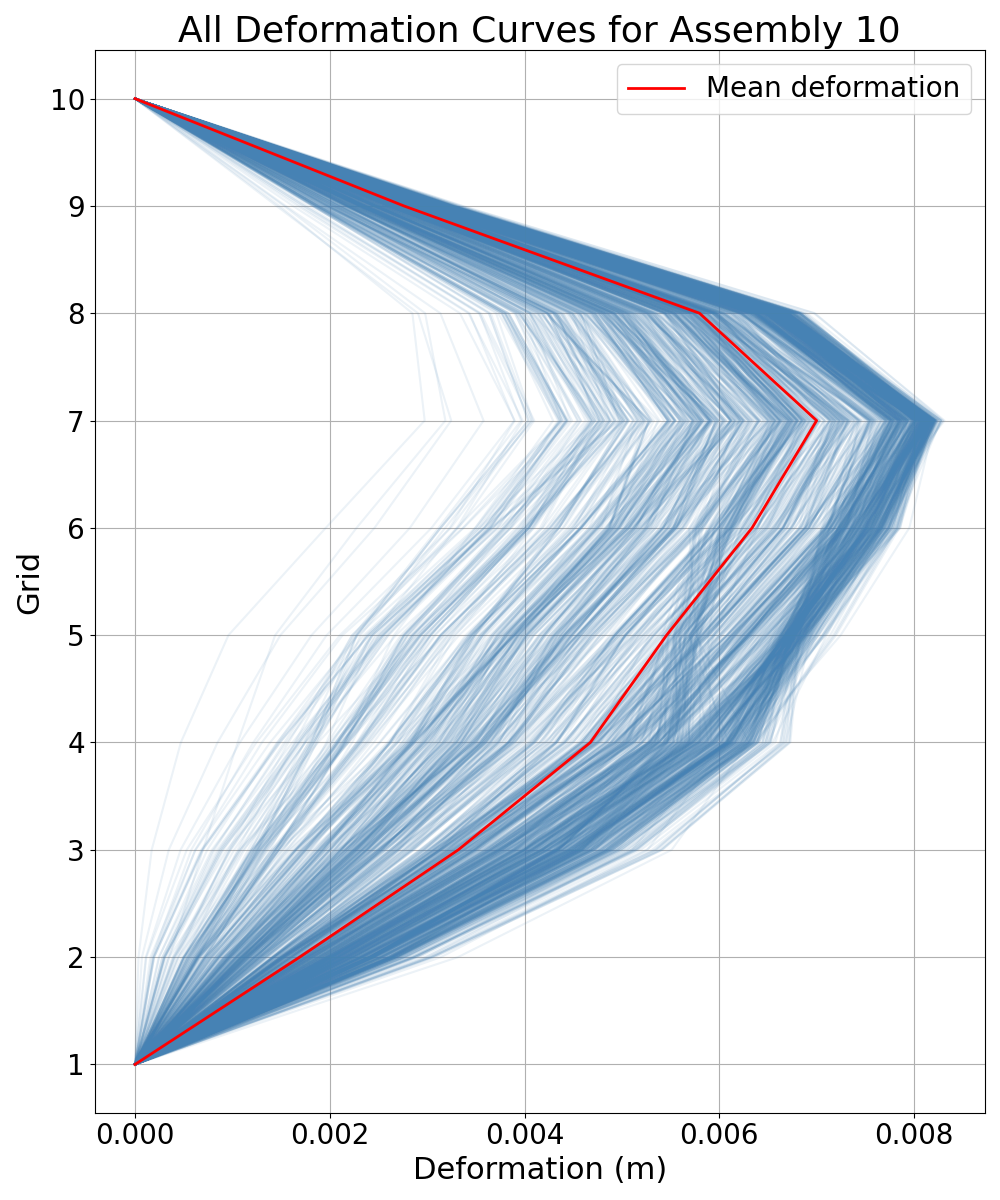}
    \end{minipage}
    \caption{A sample of size $2{,}000$ of deformation curves for assemblies $A_6$ and $A_{10}$ of cycle 3 (see Figure \ref{cycle3})}
    \label{deformationsall}
\end{figure}

%% file: conclusion.tex
This work has introduced a general framework for incorporating Gaussian process surrogates into coupled multiphysics simulations and for quantifying the induced epistemic uncertainty. At the theoretical level, the analysis combines: fill–distance bounds for scalar GPs in Sobolev-type RKHSs, lifted to matrix-valued LMC kernels, with a Lipschitz stability result for the coupled fixed point under uniform 
perturbations of the surrogates. This leads to a finite-sample, high-probability bound on the deviation between the coupled solution obtained with GP means and the random outputs of a Monte Carlo perturbation scheme, expressed in Loewner order and explicitly controlled by the fill distances of the designs and the kernel hyperparameters. The resulting bounds apply to general vector-valued surrogates and abstract coupling operators satisfying a uniform contraction property, and are therefore not tied to a specific physical model.

An analytical benchmark with a closed-form coupled solution has been used to validate the proposed constant-offset Monte Carlo scheme against a fully coherent, trajectory-conditioned sampling strategy, showing statistical agreement of the two output distributions at practical sample sizes. Finally, the methodology has been applied to a realistic fluid–structure interaction problem for fuel assembly bow over irradiation cycles, where it was shown that the uncertainty due to the coupled GP surrogates remains small (on the order of a few tenths of a millimetre) compared with the dominant input-data uncertainties. This supports the use of GP mean predictions as accurate surrogates within reactor-scale uncertainty quantification, while the theoretical results ensure that analogous conclusions can be drawn for other coupled GP-based models under the same regularity and stability assumptions.

%% file: references.bib
@phdthesis{delam2021,
  author = {de Lambert, Stanislas},
  title = {Contribution to the multiphysical analysis of fuel assembly bow},
  school = {Université Paris-Saclay},
  year = {2021},
  note = {Thèse de doctorat dirigée par Faucher, Vincent; Campioni, Guillaume; et Cardolaccia, Jérôme. Spécialité: Énergétique. Identifiant: 2021UPAST035}
}

@inproceedings{Andersson2005,
    author       = {T. Andersson and J. Almberger and L. Björnkvist},
    title        = {A Decade of Assembly Bow Management at Ringhals},
    booktitle    = {IAEA-TECDOC-1454 Structural Behaviour of Fuel Assemblies for Water Cooled Reactors},
    pages        = {129--136},
    address      = {Vienna, Austria},
    month        = {July},
    year         = {2005},
    organization = {IAEA}
}

@article{DELAMBERT2023104668,
title = {Semi-analytical fluid-structure model for the analysis of fuel assembly bow in full PWR cores},
journal = {Progress in Nuclear Energy},
year = {2023},
author = {Stanislas {de Lambert} and Jérome Cardolaccia and Vincent Faucher and Bertrand Leturcq and Guillaume Campioni}
}

@inproceedings{abboud:cea-05155171,
  title        = {Towards Uncertainty Quantification: Efficient Surrogate Models in Coupled Fluid-Structure Interaction for Fuel Assembly Bow},
booktitle    = {International Conference on Probabilistic Safety Assessment and Analysis (PSA 2025)},
  author       = {Abboud, Ali and Garnier, Josselin and Leturcq, Bertrand and Lamorte, Nicolas and de Lambert, Stanislas},
  year         = {2025},
}

@article{abboud:cea-04598002,
title = {Sensitivity analysis of a flow redistribution model for a multidimensional and multifidelity simulation of fuel assembly bow in a pressurized water reactor},
journal = {Nuclear Engineering and Design},
volume = {443},
pages = {114259},
year = {2025},
issn = {0029-5493},
doi = {https://doi.org/10.1016/j.nucengdes.2025.114259},
author = {Ali Abboud and Stanislas {de Lambert} and Josselin Garnier and Bertrand Leturcq and Nicolas Lamorte}
}

@article{Sobol,
  author = {I.M Sobol’},
  journal = {Mathematical Modeling and Computer Experiments},
  pages = {407-414},
  title = {“Sensitivity estimates for non linear mathematical models”.},
  year = {1993}
}

@article{damblin2013numerical,
  author    = {Guillaume Damblin and Marie Couplet and Bertrand Iooss},
  title     = {Numerical Studies of Space-Filling Designs: Optimization of Latin Hypercube Samples and Subprojection Properties},
  journal   = {Journal of Simulation},
  volume    = {7},
  pages     = {276--289},
  year      = {2013}
}

@book{rasmussen2006gaussian,
  title={Gaussian Processes for Machine Learning},
  author={Rasmussen, Carl Edward and Williams, Christopher KI},
  year={2006},
  publisher={MIT Press}
}

@article{aronszajn1950,
  title={Theory of reproducing kernels},
  author={Aronszajn, Nachman},
  journal={Transactions of the American Mathematical Society},
  volume={68},
  number={3},
  pages={337--404},
  year={1950},
  publisher={American Mathematical Society}
}

@book{Saltelli2000,
  title={Sensitivity Analysis},
  author={Saltelli, Andrea and Chan, Karen and Scott, E. M.},
  year={2000},
  publisher={Wiley}
}

@article{iooss,
  TITLE = {{Uncertainty and sensitivity analysis of functional risk curves based on Gaussian processes}},
  AUTHOR = {Iooss, Bertrand and Le Gratiet, Lo{\"i}c},
  JOURNAL = {{Reliability Engineering and System Safety}},
  PUBLISHER = {{Elsevier}},
  VOLUME = {187},
  PAGES = {58-66},
  YEAR = {2019}
}

@article{micchelli2005,
  author  = {Charles A. Micchelli and Massimiliano Pontil},
  title   = {On Learning Vector-Valued Functions},
  journal = {Neural Computation},
  volume  = {17},
  number  = {1},
  pages   = {177--204},
  year    = {2005},
  doi     = {10.1162/0899766052530802}
}

@article{alvarez2012,
  author    = {Mauricio A. {\'A}lvarez and Lorenzo Rosasco and Neil D. Lawrence},
  title     = {Kernels for Vector-Valued Functions: A Review},
  journal   = {Foundations and Trends in Machine Learning},
  volume    = {4},
  number    = {3},
  pages     = {195--266},
  year      = {2012},
  publisher = {Now Publishers},
  doi       = {10.1561/2200000036}
}

@inproceedings{bonilla2008,
  author    = {Edwin V. Bonilla and Kian Ming A. Chai and Christopher K. I. Williams},
  title     = {Multi-task Gaussian Process Prediction},
  booktitle = {Advances in Neural Information Processing Systems (NeurIPS)},
  year      = {2008},
  url       = {https://homepages.inf.ed.ac.uk/ckiw/postscript/multitaskGP_v22.pdf}
}

@article{blanchard:cea-02052632,
  TITLE = {{"The Uranie platform: an open-source software for optimisation, meta-modelling and uncertainty analysis"}},
  AUTHOR = {Blanchard, J.-B. and Damblin, G. and Martinez, J.-M. and Arnaud, G. and Gaudier, F.},
  JOURNAL = {{EPJ N - Nuclear Sciences \& Technologies}},
  PUBLISHER = {{EDP Sciences}},
  VOLUME = {5},
  PAGES = {4},
  YEAR = {2019},
  MONTH = Jan,
}

@inproceedings{karlsson1999modelling,
  author    = {Karlsson, L. and Manngard, T.},
  title     = {Modelling of PWR Fuel Assembly Deformations during Irradiation},
  booktitle = {Structural Mechanics in Reactor Technology},
  address   = {Seoul},
  year      = {1999}
}

@book{kerkar2008exploitation,
  author    = {Kerkar, N. and Paulin, P.},
  title     = {Exploitation des coeurs REP},
  publisher = {EDP Sciences},
  year      = {2008}
}

@phdthesis{wanningerphd,
  author = {Andreas Wanninger},
  title = {Mechanical Analysis of the Bow Deformation of Fuel Assemblies in a Pressurized Water Reactor Core},
  school = {Technische Universität München},
  year = {2018},
  address = {Munich, Germany},
  month = {July},
  type = {Doctoral Dissertation}
}

@book{zeidler1986nonlinear,
  title={Nonlinear Functional Analysis and its Applications I},
  author={Zeidler, E.},
  year={1986},
  publisher={Springer}
}

@misc{kanagawa2018,
      title={Gaussian Processes and Kernel Methods: A Review on Connections and Equivalences}, 
      author={Motonobu Kanagawa and Philipp Hennig and Dino Sejdinovic and Bharath K Sriperumbudur},
      year={2018},
      eprint={1807.02582},
      archivePrefix={arXiv},
      primaryClass={stat.ML},
      url={https://arxiv.org/abs/1807.02582}, 
}

@article{micchelli2005vv,
  author  = {Micchelli, Charles A. and Pontil, Massimiliano},
  title   = {On learning vector-valued functions},
  journal = {Neural Computation},
  year    = {2005},
  volume  = {17},
  number  = {1},
  pages   = {177--204}
}

@article{alvarez2012review,
  author  = {{\'A}lvarez, Mauricio A. and Rosasco, Lorenzo and Lawrence, Neil D.},
  title   = {Kernels for Vector-Valued Functions: A Review},
  journal = {Foundations and Trends in Machine Learning},
  year    = {2012},
  volume  = {4},
  number  = {3},
  pages   = {195--266},
  publisher = {Now Publishers}
}

@article{DELAMBERT2019330,
  title   = {Modeling the consequences of fuel assembly bowing on PWR core neutronics using a Monte-Carlo code},
  author  = {de Lambert, Stanislas and Campioni, Guillaume and Faucher, Vincent and Leturcq, Bertrand and Cardolaccia, J{\'e}r{\^o}me},
  journal = {Annals of Nuclear Energy},
  year    = {2019},
  volume  = {<volume>},
  pages   = {330--<endpage>},
  doi     = {<doi>}
}

@article{WANNINGER2018297,
title = {{"Mechanical analysis of the bow deformation of a row of fuel assemblies in a PWR core"}},
journal = {Nuclear Engineering and Technology},
author ={Wanninger, A},
year = {2018},
}

@PHDTHESIS{abboudd2025,
url = "http://www.theses.fr/2025IPPAX072",
title = "Uncertainty quantification applied to fuel assembly bow in a pressurized water reactor",
author = "Abboud, Ali",
year = "2025",
school = "Thèse de doctorat dirigée par Garnier, Josselin Mathématiques appliquées Institut polytechnique de Paris 2025"
}

@inproceedings{abboud2025,
  author    = {Ali Abboud and Josselin Garnier and Bertrand Leturcq and Julien Pacull and Olivier Fandeur and Stanislas de Lambert},
  title     = {Uncertainty Quantification Of Fuel Assembly Bow In Pressurized Water Reactor Through a Thermomechanical Simulation},
  booktitle = {Proceedings of the M\&C 2025 - The International Conference on Mathematics and Computational Methods Applied to Nuclear Science and Engineering},
  year      = {2025},
  address   = {Denver, Colorado, USA},

}
